\documentclass[runningheads,a4paper]{llncs}

\usepackage{amsmath}
\usepackage{amssymb}
\usepackage{wrapfig}
\usepackage{proof}
\usepackage[all]{xy}
\usepackage{multicol}
\usepackage{comment}
\usepackage{graphicx}
\usepackage{epstopdf}
\usepackage{url}

\newcommand{\inset}[2]{\{\, {#1} \,|\, {#2} \,\}}

\newcommand{\setof}[1]{\{\,{#1}\,\}}

\def\iff{\,\mathrm{iff}\,}

\usepackage{enumitem}\setenumerate{topsep=1pt}

\title{Axiomatizing Epistemic Logic of Friendship \\
via Tree Sequent Calculus}

\titlerunning{Axiomatizing Epistemic Logic of Friendship via Tree Sequent Calculus}
\author{Katsuhiko Sano}
\authorrunning{Sano}

\institute{Department of Philosophy, Graduate School of Letters, \\
Hokkaido University, Sapporo, Japan\\
\email{v-sano@let.hokudai.ac.jp}\\
}

\toctitle{Lecture Notes in Computer Science}
\tocauthor{Authors' Instructions}

\begin{document}

\mainmatter  

\maketitle

\begin{abstract}
This paper positively solves an open problem if it is possible to provide a Hilbert system to Epistemic Logic of Friendship (EFL) by Seligman, Girard and Liu. To find a Hilbert system, we first introduce a sound, complete and cut-free tree (or nested) sequent calculus for EFL, which is an integrated combination of Seligman's sequent calculus for basic hybrid logic and a tree sequent calculus for modal logic. Then we translate a tree sequent into an ordinary formula to specify a Hilbert system of EFL and finally show that our Hilbert system is sound and complete for an intended two-dimensional semantics. 
\end{abstract}

\section{Introduction}

Epistemic Logic of Friendship ($\mathbf{EFL}$) is a version of two-dimensional modal logic proposed by~\cite{SLG2011,SLG2013a,SLG2013}. Compared to the ordinary epistemic logic~\cite{Hintikka1962}, one of the key features of their logic is to encode the information of agents into the object language by a technique of hybrid logic~\cite{Blackburn2006,AC2007}. Then, a propositional variable $p$ can be read as an indexical proposition such as ``I am $p$'' and we may formalize the sentences like ``I know that all my friends is $p$'' or ``Each of my friends knows that he/she is $p$.'' Moreover, the authors of~\cite{SLG2013a,SLG2013} added a dynamic mechanism to $\mathbf{EFL}$ for capturing public announcements~\cite{Plaza1989}, announcements to all the friends, and private announcements~\cite{BMS98} and established a relative completeness result (cf.~\cite{SLG2013a,SLG2013,GSL2012GDDL}), i.e., they provided a set of recursion axioms for dynamic operators. So once we can provide a sound and complete proof system for $\mathbf{EFL}$, i.e., the fragment without dynamic operators, we can also establish the semantic completeness of the dynamic extension of $\mathbf{EFL}$. Therefore, this paper focuses on an open problem of axiomatizing $\mathbf{EFL}$ in terms of Hilbert system, i.e., the static part of their framework. 

A difficulty of the problem comes from a combination of {\em modal logic} for agents' knowledge and {\em hybrid logic} for a friendship relation among agents. If we combine {\em two hybrid logics} over two-dimensional semantics of~\cite{SLG2011,SLG2013a,SLG2013}, it is noted that there is an axiomatization of all valid formulas in the semantics by~\cite[p.471]{Sano2010a}. Our approach to tackle the problem is via a sequent calculus, whose idea is originally from Gentzen. In particular, our notion of sequent for $\mathbf{EFL}$ can be regarded as a combination of a tree or nested sequent~\cite{Kashima1994,Bruennler2009} for modal logic and $@$-prefixed sequent~\cite{Seligman2001,Brauner2011} for hybrid logic. One of the merits of our notion of sequent is that we can still translate our sequent into an ordinary formula. This allows us to specify our desired Hilbert system for $\mathbf{EFL}$. We note that~\cite{Christoff2016} independently provided a prefixed tableau system for a dynamic extension of $\mathbf{EFL}$. There are at least three points we should emphasize on our work. First, our tree sequent system is quite simpler than the tableau system given in~\cite{Christoff2016}, i.e., the number of rules of our sequent system is almost half of the number of rules of their system. Second, it is not clear if a prefixed formula in~\cite{Christoff2016} for the tableau calculus can be translated into an ordinary formula. Their result is not concerned with Hilbert systems. Third, their syntax contains a special kind of propositional variable (called {\em feature proposition}) and they include a tableau rule called {\em propositional cut} to handle such propositions. On the other hand, we can show that our tree sequent calculus enjoys the cut elimination theorem, the most fundamental theorem in proof-theory. 

We proceed as follows. Section \ref{sec:syn_sem} introduces the syntax and semantics of $\mathbf{EFL}$. Section \ref{sec:tree_seq} provides a tree sequent calculus for $\mathbf{EFL}$ and establishes the soundness of the sequent calculus (Theorem \ref{thm:sound_tree_seq4efl}). Section \ref{sec:comp_tree} establishes a completeness result of a cut-free fragment of our sequent calculus (Theorem \ref{thm:comp_tree_seq4efl}). As a corollary, we also provide a semantic proof of the cut elimination theorem of our sequent calculus (Theorems \ref{thm:comp_hil4efl} and \ref{thm:sound_hil4efl}, Corollary \ref{cor:sum_tree}). Section \ref{sec:hil_helf} specifies a Hilbert system of $\mathbf{EFL}$, and provides a syntactic proof of the equipollence between our proposed Hilbert system and our tree sequent calculus, which implies the soundness and completeness results for our Hilbert system (Corollary \ref{cor:summary}). Section \ref{sec:ext} extends our technical results to cover extensions of $\mathbf{EFL}$ where a modal operator for states (or a knowledge operator) obeys $\mathbf{KT}$, $\mathbf{S4}$ or $\mathbf{S5}$ axioms and a friendship relation satisfies a certain form of universal property (Theorems \ref{thm:ex_tree} and \ref{thm:ex_hil}, Corollary \ref{cor:comp_efls5}). The result of this section subsumes the logic given in~\cite{Christoff2016}, provided we drop the dynamic operator from the syntax of~\cite{Christoff2016}. Section \ref{sec:fd} concludes this paper. 

\section{Syntax and Two-dimensional Kripke Semantics}
\label{sec:syn_sem}

Our syntax $\mathcal{L}$ consists of the following vocabulary: a countably infinite set $\mathsf{Prop}$ = $\setof{p,q,r,\ldots}$ of propositional variables, a countably infinite set $\mathsf{Nom}$ = $\setof{n,m,l,\ldots}$ of agent nominal variables, the Boolean connectives of $\to$ (the implication) and $\bot$ (the falsum), the satisfaction operators $@$ and the friendship operator $\mathsf{F}$ (read as ``all my friends are ...'') as well as the modal operator $\Box$ which may be regarded as the knowledge operator. We note that an {\em agent nominal} $n \in \mathsf{Nom}$ is a syntactic name of an agent or an individual, which amounts to a constant symbol of the first-order logic, while $n$ is read indexically as ``I am $n$.'' Similarly, we read a propositional variable $p \in \mathsf{Prop}$ also indexically by ``I am $p$,'' e.g., ``I am in danger.'' The set $\mathsf{Form}$ of formulas in $\mathcal{L}$ is defined inductively as follows:
\[
\mathsf{Form} \ni \varphi :: =  n \,|\, p \,|\, \bot \,|\, \varphi \to \varphi \,|\, @_{n}\varphi \,|\, \mathsf{F} \varphi \,|\, \Box \varphi,
\]
where $n \in \mathsf{Nom}$ and $p \in \mathsf{Prop}$. Boolean connectives other than $\to$ or $\bot$ are introduced as ordinary abbreviations. We define the dual of $\Box$ as $\Diamond$ := $\neg \Box \neg$ and the dual of $\mathsf{F}$ as $\langle \mathsf{F} \rangle$  := $\neg \mathsf{F} \neg$. Moreover, a formula of the form $@_{n} \varphi$ is said to be {\em $@$-prefixed}. Let us read $\Box$ as ``I know that.'' Here are some examples of how to read formulas:
\begin{itemize}
\item $\Box p$, read as ``I know that I am $p$.''
\item $@_{n}\Box p$, read as ``$n$ knows that she is $p$.'' 
\item $\Box @_{n}p$, read as ``I know that agent $n$ is $p$.''
\item $\mathsf{F} p$, read as ``all my friends are $p$.''
\item $\mathsf{F} \Box p$, read as ``all my friends know that they are $p$.''
\item $\Box \mathsf{F} p$, read as ``I know that all my friends are $p$.''
\item $@_{n} \langle \mathsf{F}\rangle m$, read as ``agent $m$ is a friend of agent $n$.''
\end{itemize}
We say that a mapping $\sigma: \mathsf{Prop} \cup \mathsf{Nom} \to \mathsf{Form}$ is a {\em uniform substitution} if $\sigma$ uniformly substitutes propositional variables by formulas and agent nominals by agent nominals and we use $\varphi \sigma$ to mean the result of applying a uniform substitution $\sigma$ to $\varphi$. In particular, we use $\varphi[n/k]$ to mean the result of substituting each occurrence of agent nominal $k$ in $\varphi$ uniformly with agent nominal $n$. 

A {\em model} $\mathfrak{M}$ for our syntax $\mathcal{L}$ is a tuple 
\[
(W, A, (R_{a})_{a \in A}, (\asymp_{w})_{w \in W},V),
\] 
\noindent where $W$ is a non-empty set of possible states, $A$ is a non-empty set of agents, $R_{a}$ is a binary relation on $W$ ($a \in A$), $\asymp_{w}$ is a binary relation on $A$ (called a {\em friendship relation}, $w \in W$), $V$ is a valuation function $\mathsf{Prop} \cup \mathsf{Nom} \to \mathcal{P}(W \times A)$ such that $V(n)$ is a subset of $W \times A$ of the form $W \times \setof{a}$. When $V(n)$ = $W \times \setof{a}$, we denote such unique element $a$ by $\underline{n}$. We note that a semantic value $\underline{n}$ of a nominal $n$ is {\em rigid} over all possible states. We do not require any property for $R_{a}$ and $\asymp_{w}$ but we will come back to this point in Section \ref{sec:ext}. We say that a tuple $\mathfrak{F}$ = $(W, A, (R_{a})_{a \in A}, (\asymp_{w})_{w \in W})$ without a valuation is a {\em frame}. 

Let $\mathfrak{M}$ = $(W, A, (R_{a})_{a \in A}, (\asymp_{w})_{w \in W},V)$ be a model. Given a pair $(w,a) \in W \times A$ and a formula $\varphi$, the satisfaction relation $\mathfrak{M},(w,a) \models \varphi$ (read ``agent $a$ satisfies $\varphi$ at $w$ in $\mathfrak{M}$ '') inductively as follows:
\[
\begin{array}{lll}
\mathfrak{M}, (w,a) \models p & \iff & (w,a) \in V(p), \\
{\mathfrak{M}}, (w,a) \models n & \iff & \underline{n} = a,  \\
{\mathfrak{M}}, (w,a) \not\models \bot &  &   \\
{\mathfrak{M}}, (w,a) \models \varphi \to \psi & \iff &  {\mathfrak{M}}, (w,a) \models \varphi \text{ implies } {\mathfrak{M}}, (w,a) \models \psi \\
{\mathfrak{M}}, (w,a) \models @_{n} \varphi & \iff & \mathfrak{M}, (w,\underline{n}) \models \varphi, \\
\mathfrak{M}, ({w},a) \models  \mathsf{F} \varphi & \iff & (a\asymp_{{w}}b \text{ implies }\mathfrak{M}, ({w},b) \models \varphi) \text{ for all agents $b \in A$},  \\
\mathfrak{M}, (w,{a}) \models \Box \varphi & \iff& 
(wR_{{a}}v \text{ implies }\mathfrak{M}, (v,{a}) \models \varphi) \text{ for all states $v \in W$}.  \\
\end{array}
\]
Given a class $\mathbb{M}$ of models, we say that a formula $\varphi$ is {\em valid} in $\mathbb{M}$ when $\mathfrak{M}, (w,a) \models \varphi$ for all pairs $(w,a)$ in $\mathfrak{M}$ and all models $\mathfrak{M} \in \mathbb{M}$. This paper tackles the question if the set of all valid formulas in the class of all models is axiomatizable.

\section{Tree Sequent Calculus of Epistemic Logic of Friendship}
\label{sec:tree_seq}
A {\em label} is inductively defined as follows: Any natural number is a label; if $\alpha$ is a label, $n$ is an agent nominal in $\mathsf{Nom}$ and $i$ is a natural number, then $\alpha \cdot_{n} i$ is also a label. When $\beta$ is $\alpha \cdot_{n} i$, then we say that $\beta$ is an {\em $n$-child} of $\alpha$ or that $\alpha$ is {\em an $n$-parent} of $\beta$. A {\em tree} $\mathcal{T}$ is a set of labels such that the set contains the unique natural number $j$ as the root label and the set is closed under taking the parent of a label, i.e., $\alpha \cdot_{n} i \in \mathcal{T}$ implies $\alpha \in \mathcal{T}$ for all labels $\alpha$, agent nominals $n$ and natural numbers $i$. For example, all of $0$, $0\cdot_{n}1$ and $0\cdot_{k} 2$ are labels and they form a finite tree. 

\begin{figure}[htbp]
\begin{center}
{\includegraphics[clip, width=50mm]{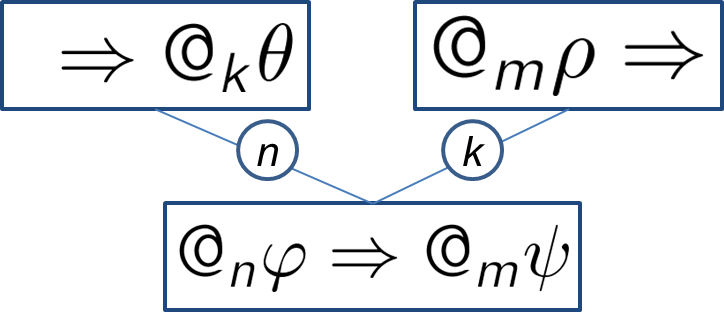}}\\
\end{center}
\caption{A tree sequent}
\label{fig:tree_seq}
\end{figure}

Given a label $\alpha$ and an {\em $@$-prefixed formula} $\varphi$, the expression $\alpha: \varphi$ is said to be a {\em labelled formula}, where recall that an $@$-prefixed formula is of the form $@_{n}\varphi$. 
A {\em tree sequent} is an expression of the form 
\begin{center}
$\Gamma \stackrel{\mathcal{T}}{\Rightarrow} \Delta$
\end{center}
where $\Gamma$ and $\Delta$ are {\em finite} sets of labelled formulas, $\mathcal{T}$ is a finite tree of labels, and all the labels in $\Gamma$ and $\Delta$ are in $\mathcal{T}$. A tree sequent ``$\Gamma \stackrel{\mathcal{T}}{\Rightarrow} \Delta$'' is read as ``if we assume all labelled formulas in $\Gamma$, then we may conclude some labelled formulas in $\Delta$.'' A tree sequent $0:@_{n}\varphi, 0\cdot_{k}2: @_{m} \rho \stackrel{\mathcal{T}}{\Rightarrow} 0:@_{m}\psi, 0\cdot_{n}1: @_{k} \theta$ is represented as in Fig.~\ref{fig:tree_seq}, where $\mathcal{T}$ = $\setof{0, 0\cdot_{n}1, 0\cdot_{k} 2}$. That is,  $0$, $0\cdot_{n}1$ and $0\cdot_{k} 2$ are ``addresses'' of the root, the left leaf, and the right leaf, respectively. Therefore, our tree sequent is a finite tree, each of which nodes has an $@$-prefixed sequent as given in~\cite{Seligman2001,Brauner2011}.

\begin{table}[htbp]
\hrule
\vspace{1pt}
\caption{Tree Sequent Calculus $\mathsf{T}\mathbf{EFL}$}
\label{table:tree_seq_calc}
\begin{tabular}{cc}
$(\bot)$ \,\, $\alpha: @_{n}\bot, \Gamma\stackrel{\mathcal{T}}{\Rightarrow} \Delta$ & 
$(\mathsf{id})$ \,\,${\alpha: @_{n}\varphi}, \Gamma \stackrel{\mathcal{T}}{\Rightarrow} \Delta,  {\alpha: @_{n}\varphi}$ \\
$\infer[(\mathsf{rep}_{=1})]{{\alpha:@_{n}m}, \alpha:\varphi[m/k], \Gamma  \stackrel{\mathcal{T}}{\Rightarrow} \Delta }{{\alpha:@_{n}m}, \alpha:\varphi[n/k], \Gamma  \stackrel{\mathcal{T}}{\Rightarrow} \Delta}$ & 
$\infer[(\mathsf{rep}_{=2})]{\alpha:@_{n}m, \alpha:\varphi[n/k], \Gamma  \stackrel{\mathcal{T}}{\Rightarrow} \Delta }{\alpha:@_{n}m, \alpha:\varphi[m/k],\Gamma  \stackrel{\mathcal{T}}{\Rightarrow} \Delta}$ \\ 
$\infer[(\mathsf{ref}_{=})]{\Gamma \stackrel{\mathcal{T}}{\Rightarrow} \Delta}{ {\alpha: @_{n}n}, \Gamma \stackrel{\mathcal{T}}{\Rightarrow} \Delta}$
&
$\infer[(\mathsf{rigid}_{=})]{{\alpha: @_{n}m}, \Gamma \stackrel{\mathcal{T}}{\Rightarrow} \Delta}{{\beta: @_{n}m}, \Gamma \stackrel{\mathcal{T}}{\Rightarrow} \Delta }
$ \\
$
\infer[(\mathop{\to} R)]{\Gamma \stackrel{\mathcal{T}}{\Rightarrow} \Delta, \alpha:@_{n}(\varphi \to \psi)}{\alpha:@_{n}\varphi, \Gamma \stackrel{\mathcal{T}}{\Rightarrow} \Delta, \alpha:@_{n} \psi}$
&
$\infer[(\mathop{\to} L)]{\alpha:@_{n}(\varphi \to \psi), \Gamma \stackrel{\mathcal{T}}{\Rightarrow} \Delta}{\Gamma \stackrel{\mathcal{T}}{\Rightarrow} \Delta, \alpha:@_{n}\varphi  & \alpha:@_{n}\psi, \Gamma \stackrel{\mathcal{T}}{\Rightarrow} \Delta}
$ \\
$\infer[(@R)]{\Gamma \stackrel{\mathcal{T}}{\Rightarrow} \Delta, \alpha:@_{n}@_{m}\varphi}{\Gamma \stackrel{\mathcal{T}}{\Rightarrow} \Delta, \alpha:@_{m}\varphi}$ 
&
$\infer[(@L)]{\alpha:@_{n}@_{m}\varphi, \Gamma \stackrel{\mathcal{T}}{\Rightarrow} \Delta}{\alpha:@_{m}\varphi, \Gamma \stackrel{\mathcal{T}}{\Rightarrow} \Delta}$ \\
$\infer[(\mathsf{F}R)^{\ast}]{\Gamma \stackrel{\mathcal{T}}{\Rightarrow} \Delta, {\alpha:@_{n} \mathsf{F} \varphi}}{{\alpha:@_{n} \langle \mathsf{F} \rangle m}, \Gamma \stackrel{\mathcal{T}}{\Rightarrow} \Delta, {\alpha:@_{m}\varphi}}$
&
$\infer[(\mathsf{F}L)]{{\alpha:@_{n} \mathsf{F} \varphi}, \Gamma \stackrel{\mathcal{T}}{\Rightarrow} \Delta}{
\Gamma \stackrel{\mathcal{T}}{\Rightarrow} \Delta, {\alpha:@_{n} \langle \mathsf{F} \rangle m} 
&
{\alpha:@_{m}\varphi}, \Gamma \stackrel{\mathcal{T}}{\Rightarrow} \Delta
}
$\\
$
\infer[(\Box R)^{\dagger}]{\Gamma \stackrel{\mathcal{T}}{\Rightarrow} \Delta, {\alpha:@_{n}\Box \varphi}}{\Gamma \stackrel{\mathcal{T} \cup {\setof{\gamma}}}{\Rightarrow} \Delta , {\gamma:@_{n} \varphi}}$
&
$
\infer[(\Box L)^{\ddagger}]{{\alpha:@_{n}\Box \varphi}, \Gamma \stackrel{\mathcal{T}}{\Rightarrow} \Delta}{{\beta:@_{n} \varphi}, \Gamma \stackrel{\mathcal{T}}{\Rightarrow} \Delta}
$\\
$\infer[(w \mathsf{lab})^\star]{\Gamma \stackrel{\mathcal{T} \cup \{\alpha\} }{\Rightarrow} \Delta}{\Gamma \stackrel{\mathcal{T}}{\Rightarrow} \Delta}$
&
$\infer[(Cut)]{\Gamma,\Pi \stackrel{\mathcal{T}}{\Rightarrow} \Delta, \Sigma}{\Gamma \stackrel{\mathcal{T}}{\Rightarrow} \Delta, \alpha:@_{n}\varphi &  \alpha:@_{n}\varphi, \Pi \stackrel{\mathcal{T}}{\Rightarrow} \Sigma}$
\\
\end{tabular}
\\
 $\ast$: $m$ is a fresh agent nominal in the lower sequent; $\dagger$: $\gamma$ is an $n$-child of $\alpha$ which is fresh in the lower sequent; $\ddagger$: $\beta$ is an $n$-child of $\alpha$; $\star$: \text{ $\mathcal{T} \cup \{\alpha\}$ is a tree of labels. }
\vspace{1pt} 
\hrule
\end{table}

Table \ref{table:tree_seq_calc} provides all the initial sequents and all the inference rules of tree sequent calculus $\mathsf{T}\mathbf{EFL}$, where recall that $\varphi[m/k]$ is the result of substituting each occurrence of agent nominal $k$ in $\varphi$ with agent nominal $m$. The system without the cut rule is denoted by $\mathsf{T}\mathbf{EFL}^{-}$. All the initial sequents and inference rules except $(\mathsf{rigid}_{=})$, $(w \mathsf{lab})$, $(\Box R)$ and $(\Box L)$ originate from sequent calculus for hybrid logic in terms of $@$-prefixed sequents (cf.~\cite{Seligman2001,Brauner2011} ). The inference rules $(w \mathsf{lab})$, $(\Box R)$ and $(\Box L)$ reflect the idea of tree or nested sequent calculus for modal logic (cf.~\cite{Kashima1994,Bruennler2009}). Finally the rule $(\mathsf{rigid}_{=})$ encodes the semantic idea that a semantic value of a nominal is rigid, i.e., the same through all possible states.

A {\em derivation} in $\mathsf{T}\mathbf{EFL}$ (or $\mathsf{T}\mathbf{EFL}^{-}$) is a finite tree generated from initial sequents by inference rules of $\mathsf{T}\mathbf{EFL}$ (or $\mathsf{T}\mathbf{EFL}^{-}$, respectively). The {\em height} of a derivation is defined as the maximum length of branches in the derivation from the end (or root) sequent to an initial sequent. A tree sequent $\Gamma \stackrel{\mathcal{T}}{\Rightarrow} \Delta$ is said to be {\em provable} in $\mathsf{T}\mathbf{EFL}$ (or $\mathsf{T}\mathbf{EFL}^{-}$) if there is a derivation in $\mathsf{T}\mathbf{EFL}$ (or $\mathsf{T}\mathbf{EFL}^{-}$, respectively) such that the root of the tree is $\Gamma \stackrel{\mathcal{T}}{\Rightarrow} \Delta$.

Let $\mathfrak{M}$ = $(W, A, (R_{a})_{a \in A}, (\asymp_{w})_{w \in W},V)$ be a model and $\mathcal{T}$ a tree of labels. A function $f: \mathcal{T} \to W$ is a {\em $\mathcal{T}$-assignment} in $\mathfrak{M}$ if, whenever $\beta$ is an $n$-child of $\alpha$ in $\mathcal{T}$, $f(\alpha) R_{\underline{n}} f(\beta)$ holds. When it is clear from the context, we often drop ``$\mathcal{T}$-'' from ``$\mathcal{T}$-assignment''. Given any labelled formula $\alpha:@_{n}\varphi$ with $\alpha \in \mathcal{T}$ and any $\mathcal{T}$-assignment in $\mathfrak{M}$, we define the {\em satisfaction} for a labelled formula as follows:
\[
\begin{array}{lll}
\mathfrak{M},f \models  \alpha:@_{n}\varphi & \iff&  \mathfrak{M},  ({f}(\alpha), \underline{n}) \models \varphi. \\
\end{array}
\]
where ``$\mathfrak{M},f \models  \alpha:@_{n}\varphi$'' is read as ``$\alpha:@_{n}\varphi$ is true at $(\mathfrak{M},f)$''.  Given a tree sequent $\Gamma \stackrel{\mathcal{T}}{\Rightarrow} \Delta$ and a $\mathcal{T}$-assignment in $\mathfrak{M}$, we say that $\Gamma \stackrel{\mathcal{T}}{\Rightarrow} \Delta$ is {\em true} in $(\mathfrak{M}, f)$ (notation: $\mathfrak{M},f \models \Gamma \stackrel{\mathcal{T}}{\Rightarrow} \Delta$) if, whenever all labelled formulas of $\Gamma$ is true in $(\mathfrak{M}, f)$, some labelled formulas of $\Delta$ is true in $(\mathfrak{M}, f)$. The following theorem is easy to establish. 

\begin{theorem}[Soundness of $\mathsf{T}\mathbf{EFL}$]
\label{thm:sound_tree_seq4efl}
If a tree sequent $\Gamma \stackrel{\mathcal{T}}{\Rightarrow} \Delta$ is provable in $\mathsf{T}\mathbf{EFL}$ then $\mathfrak{M},f \models \Gamma \stackrel{\mathcal{T}}{\Rightarrow} \Delta$ for all models $\mathfrak{M}$ and all assignments $f$. 
\end{theorem}


Let us say that an inference rule is {\em height-preserving admissible} in $\mathsf{T}\mathbf{EFL}^{-}$ (or $\mathsf{T}\mathbf{EFL}$) if, whenever all uppersequents (premises) of the inference rule is provable by derivations with height no more than $n$, then the lowersequent (conclusion) of the rule is provable by a derivation whose height is at most $n$. By induction on height $n$ of a derivation, we can prove the following. 

\begin{proposition}
\label{prop:admissible}
\begin{itemize}
\item[$\mathrm{(i)}$] The following substitution rule $(\mathsf{sub})$ is height-preserving admissible in $\mathsf{T}\mathbf{EFL}^{-}$ and $\mathsf{T}\mathbf{EFL}$: 
\[
\infer[(\mathsf{sub})]{\Gamma\sigma \stackrel{\mathcal{T}\sigma}{\Rightarrow} \Delta\sigma}{\Gamma \stackrel{\mathcal{T}}{\Rightarrow} \Delta},
\]
where $\sigma$ is a uniform substitution, $\mathcal{T}\sigma$ is the resulting tree by substituting agent nominals in $\mathcal{T}$ by $\sigma$, $\Theta\sigma$ $:=$ $\inset{\alpha \sigma: \varphi\sigma }{\alpha:\varphi \in \Theta}$ and $\alpha \sigma \in \mathcal{T} \sigma$ is the corresponding label to $\alpha \in \mathcal{T}$ by $\sigma$. 

\item[$\mathrm{(ii)}$] The following weakening rules $(\mathsf{w}R)$ and $(\mathsf{w}L)$ are height-preserving admissible in $\mathsf{T}\mathbf{EFL}^{-}$ and $\mathsf{T}\mathbf{EFL}$. 
\[
\infer[(\mathsf{w}R)]{\Gamma \stackrel{\mathcal{T}}{\Rightarrow}  \Delta, \alpha:@_{n}\varphi}{\Gamma \stackrel{\mathcal{T}}{\Rightarrow}  \Delta},
\quad
\infer[(\mathsf{w}L)]{\alpha:@_{n}\varphi, \Gamma \stackrel{\mathcal{T}}{\Rightarrow}  \Delta}{\Gamma \stackrel{\mathcal{T}}{\Rightarrow}  \Delta}.
\]
\end{itemize}
\end{proposition}

\section{Semantic Completeness of Tree Sequent Calculus of Epistemic Logic of Friendship}
\label{sec:comp_tree}

In what follows in this section, sets $\Gamma$, $\Delta$, etc. of labelled formulas and a tree $\mathcal{T}$ of labels can be possibly (countably) infinite. Following this change, we say that a possibly infinite tree-sequent $\Gamma \stackrel{\mathcal{T}}{\Rightarrow} \Delta$ is provable in $\mathsf{T}\mathbf{EFL}^{-}$ if there exist finite sets $\Gamma' \subseteq \Gamma$ and $\Delta' \subseteq \Delta$ and finite subtree $\mathcal{T}'$ of $\mathcal{T}$ such that $\Gamma' \stackrel{\mathcal{T}'}{\Rightarrow} \Delta'$ is provable in $\mathsf{T}\mathbf{EFL}^{-}$. 

\begin{definition}[Saturated tree sequent]
\label{dfn:saturation}
A possibly infinite tree sequent $\Gamma \stackrel{\mathcal{T}}{\Rightarrow} \Delta$ is {\em saturated} if it satisfies the following conditions:
\begin{description}
\item[(rep1)] If $\alpha:@_{n}m \in \Gamma$ and $\alpha: \varphi[n/k] \in \Gamma$ then $\alpha: \varphi[m/k] \in \Gamma$. 
\item[(rep2)] If $\alpha:@_{m}n \in \Gamma$ and $\alpha: \varphi[n/k] \in \Gamma$ then $\alpha: \varphi[m/k] \in \Gamma$. 

\item[(ref$_{=}$)] $\alpha:@_{n}n \in \Gamma$ for all labels $\alpha \in \mathcal{T}$. 
\item[(rigid$_{=}$)] If $\alpha:@_{n}m \in \Gamma$ then $\beta:  @_{n}m \in \Gamma$ for all labels $\beta \in \mathcal{T}$. 
\item[($\to$r)] If $\alpha:@_{n}(\varphi \to \psi) \in \Delta$ then $\alpha:@_{n} \varphi \in \Gamma$ and $\alpha:@_{n} \psi \in \Delta$. 
\item[($\to$l)] If $\alpha:@_{n}(\varphi \to \psi) \in \Gamma$ then $\alpha:@_{n}\varphi \in \Delta$ or $\alpha:@_{n} \psi \in \Gamma$. 
\item[($@$r)] If $\alpha:@_{n}@_{m}\varphi \in \Delta$ then $\alpha:@_{m}\varphi \in \Delta$.
\item[($@$l)] If $\alpha:@_{n}@_{m}\varphi \in \Gamma$ then $\alpha:@_{m}\varphi \in \Gamma$.
\item[($\mathsf{F}$r)] If $\alpha:@_{n}\mathsf{F}\varphi \in \Delta$ then 
$\alpha:@_{n} \langle \mathsf{F} \rangle m \in \Gamma$ and $\alpha:@_{m}\varphi \in \Delta$ for some agent nominal $m$. 
\item[($\mathsf{F}$l)] If $\alpha:@_{n}\mathsf{F}\varphi \in \Gamma$ then $\alpha:@_{n}\langle \mathsf{F}\rangle m \in \Delta$ or $\alpha: @_{m}\varphi \in \Gamma$ for all agent nominals $m$. 
\item[($\Box$r)]  If $\alpha:@_{n}\Box \varphi \in \Delta$ then $\beta:@_{n}\varphi \in \Delta$ for some $n$-child $\beta$ of $\alpha$. 
\item[($\Box$l)]  If $\alpha:@_{n}\Box \varphi \in \Gamma$ then $\beta:@_{n}\varphi \in \Gamma$ for all $n$-children $\beta$ of $\alpha$. 
\end{description}
\end{definition}


\begin{lemma}[Saturation lemma]
\label{lem:saturation}
Let $\Gamma \stackrel{\mathcal{T}}{\Rightarrow} \Delta$ be an unprovable tree sequent in $\mathsf{T}\mathbf{EFL}^{-}$. Then, there exists a saturated $($possibly infinite$)$ sequent $\Gamma^{+} \stackrel{\mathcal{T}^{+}}{\Rightarrow} \Delta^{+}$ such that it is still unprovable in $\mathsf{T}\mathbf{EFL}^{-}$ and it extends the original tree sequent, i.e., $\Gamma \subseteq \Gamma^{+}$, $\Delta \subseteq \Delta^{+}$ and $\mathcal{T} \subseteq \mathcal{T}^{+}$. 
\end{lemma}

\begin{proof}
Let $\Gamma \stackrel{\mathcal{T}}{\Rightarrow} \Delta$ be an unprovable tree sequent in $\mathsf{T}\mathbf{EFL}^{-}$. Let $(\alpha_{i}: @_{{n}_{i}} \varphi_{i})_{i \in \omega}$ be an enumeration of all labelled formulas such that each labelled formula occurs infinitely often. In what follows, we inductively define a sequence $(\Gamma_{i} \stackrel{\mathcal{T}_{i}}{\Rightarrow} \Delta_{i})_{i \in \omega}$ of unprovable tree sequents in $\mathsf{T}\mathbf{EFL}^{-}$ such that $\Gamma_{i} \subseteq \Gamma_{i+1}$, $\Delta_{i} \subseteq \Delta_{i+1}$ and $\mathcal{T}_{i} \subseteq \mathcal{T}_{i+1}$ for all $i \in \omega$. 
\noindent \textbf{(Basis)} When $i=0$, a tree sequent $\Gamma_{0} \stackrel{\mathcal{T}_{0}}{\Rightarrow} \Delta_{0}$ is defined as the tree sequent $\Gamma \stackrel{\mathcal{T}}{\Rightarrow} \Delta$ which is clearly unprovable in $\mathsf{T}\mathbf{EFL}^{-}$. \\
\noindent \textbf{(Inductive Step)} Suppose that we have defined $(\Gamma_{i} \stackrel{\mathcal{T}_{i}}{\Rightarrow} \Delta_{i})_{0 \leqslant i \leqslant j}$. We define $\Gamma_{j+1} \stackrel{\mathcal{T}_{j+1}}{\Rightarrow} \Delta_{j+1}$ in the following two steps. 
\begin{description}
\item[Step 1: ] This step expands $\Gamma_{j}$ by the rules $(\mathsf{rep}_{=1})$, $(\mathsf{rep}_{=2})$, $(\mathsf{ref}_{=})$ and $(\mathsf{rigid}_{=})$ while $\Delta_{j}$ and $\mathcal{T}_{j}$ are unchanged. First, we enumerate all the finite pairs of the form 
\begin{center}
$(\alpha:@_{n}m, \alpha: \varphi[n/k])$ or $(\alpha:@_{m}n, \alpha: \varphi[n/k])$ 
\end{center}
found in $\Gamma_{j}$ and for each such pair we add $\alpha: \varphi[m/k]$ to $\Gamma_{j}$ to define the expanded set as $\Gamma_{j}^{\mathsf{rep}}$. It is easy to see that $\Gamma_{j}^{\mathsf{rep}} \stackrel{\mathcal{T}_{j}}{\Rightarrow} \Delta_{j}$ is unprovable in $\mathsf{T}\mathbf{EFL}^{-}$ by $(\mathsf{rep}_{=1})$ and $(\mathsf{rep}_{=2})$. 
Second, we define 
\[
\Gamma_{j}^{\mathsf{ref}} := \inset{\alpha: @_{n}n}{\alpha \in \mathcal{T}_{j} \text{ and $n$ occurs in $\Gamma_{j}^{\mathsf{rep}} \stackrel{\mathcal{T}_{j}}{\Rightarrow} \Delta_{j}$}}. 
\] 
It is immediate to see that $\Gamma_{j}^{\mathsf{ref}} \stackrel{\mathcal{T}_{j}}{\Rightarrow} \Delta_{j}$ is unprovable in $\mathsf{T}\mathbf{EFL}^{-}$ by $(\mathsf{ref}_{=})$. 
Finally we define $\Gamma_{j}^{=}$ := $\inset{\beta: @_{n}m}{\alpha: @_{n}m \in \Gamma_{j}^{\mathsf{ref}} \text{ and } \beta \in \mathcal{T}_{j}}$. 
Then the unprovability of $\Gamma_{j}^{=} \stackrel{\mathcal{T}_{j}}{\Rightarrow} \Delta_{j}$ in $\mathsf{T}\mathbf{EFL}^{-}$ is due to $(\mathsf{rigid}_{=})$. We note that $\Gamma_{j}^{=}$ is still finite.

\item[Step 2: ]  This step expands the unprovable tree sequent $\Gamma_{j}^{=} \stackrel{\mathcal{T}_{j}}{\Rightarrow} \Delta_{j}$ by logical rules, depending on the form of the $j$-th element $\alpha_{j}: @_{n_{j}} \varphi_{j}$ of our enumeration of labelled formulas. 

\begin{itemize}

\item Let $\varphi_{j}$ be of the form $\psi_{1} \to \psi_{2}$ and $\alpha_{j}: @_{n_{j}} (\psi_{1} \to \psi_{2}) \in \Gamma_{j}^{=}$. 
Then either 
\begin{center}
$\Gamma_{j}^{=} \stackrel{\mathcal{T}_{j}}{\Rightarrow} \Delta_{j}, \alpha_{j}: @_{n_{j}} \psi_{1}$ or $\alpha_{j}: @_{n_{j}} \psi_{2}, \Gamma_{j}^{=} \stackrel{\mathcal{T}_{j}}{\Rightarrow} \Delta_{j}$ 
\end{center}
is unprovable in $\mathsf{T}\mathbf{EFL}^{-}$ by $(\to L)$. We choose an unprovable tree sequent as $\Gamma_{j+1} \stackrel{\mathcal{T}_{j+1}}{\Rightarrow} \Delta_{j+1}$. 

\item Let $\varphi_{j}$ be of the form $\psi_{1} \to \psi_{2}$ and $\alpha_{j}: @_{n_{j}} (\psi_{1} \to \psi_{2}) \in \Delta_{j}$. Then 
\begin{center}
$\alpha_{j}: @_{n_{j}} \psi_{1}, \Gamma_{j}^{=} \stackrel{\mathcal{T}_{j}}{\Rightarrow} \Delta_{j}, \alpha_{j}: @_{n_{j}} \psi_{2}$ 
\end{center}
is unprovable in $\mathsf{T}\mathbf{EFL}^{-}$ by $(\to R)$ and it is chosen as $\Gamma_{j+1} \stackrel{\mathcal{T}_{j+1}}{\Rightarrow} \Delta_{j+1}$. 

\item Let $\varphi_{j}$ be of the form $@_{m}\psi$ and $\alpha_{j}: @_{n_{j}} @_{m}\psi \in \Gamma_{j}^{=}$. Then 
\begin{center}
$\alpha_{j}: @_{m} \psi_{1}, \Gamma_{j}^{=} \stackrel{\mathcal{T}_{j}}{\Rightarrow} \Delta_{j}$ 
\end{center}
is unprovable in $\mathsf{T}\mathbf{EFL}^{-}$ by $(@ L)$ and it is chosen as $\Gamma_{j+1} \stackrel{\mathcal{T}_{j+1}}{\Rightarrow} \Delta_{j+1}$. 

\item Let $\varphi_{j}$ be of the form $@_{m}\psi$ and $\alpha_{j}: @_{n_{j}} @_{m}\psi \in \Delta_{j}$. Then 
\begin{center}
$\Gamma_{j}^{=} \stackrel{\mathcal{T}_{j}}{\Rightarrow} \Delta_{j}, \alpha_{j}: @_{m} \psi_{1}$
\end{center}
is unprovable in $\mathsf{T}\mathbf{EFL}^{-}$ by $(@ R)$ and it is chosen as $\Gamma_{j+1} \stackrel{\mathcal{T}_{j+1}}{\Rightarrow} \Delta_{j+1}$. 

\item Let $\varphi_{j}$ be of the form $\mathsf{F}\psi$ and $\alpha_{j}: @_{n_{j}} \mathsf{F}\psi \in \Gamma_{j}^{=}$. Let $m_{0}$, $\ldots$, $m_{l}$ be all the finite agent nominals occuring in $\Gamma_{j}^{=} \stackrel{\mathcal{T}_{j}}{\Rightarrow} \Delta_{j}$. 
We define an increasing sequence $(\Gamma_{j}^{(i)} \stackrel{\mathcal{T}_{j}}{\Rightarrow} \Delta_{j}^{(i)})_{0 \leqslant i \leqslant l+1}$ of unprovable tree sequent in $\mathsf{T}\mathbf{EFL}^{-}$ as follows (it is noted that $\mathcal{T}_{j}$ is unchanged in this process). 
We put $\Gamma_{j}^{(i)} \stackrel{\mathcal{T}_{j}}{\Rightarrow} \Delta_{j}^{(i)}$ := $\Gamma_{j}^{=} \stackrel{\mathcal{T}_{j}}{\Rightarrow} \Delta_{j}$. Suppose that we have constructed $(\Gamma_{j}^{(i)} \stackrel{\mathcal{T}_{j}}{\Rightarrow} \Delta_{j}^{(i)})_{1 \leqslant i \leqslant h}$. Then either 
\begin{center}
$\Gamma_{j}^{(h)} \stackrel{\mathcal{T}_{j}}{\Rightarrow} \Delta_{j}^{(h)}, \alpha_{j}: @_{n_{j}}\langle \mathsf{F} \rangle m_{h}$ or $\alpha_{j}:  @_{m_{h}} \psi, \Gamma_{j}^{(h)} \stackrel{\mathcal{T}_{j}}{\Rightarrow} \Delta_{j}^{(h)}$ 
\end{center}
is unprovable in $\mathsf{T}\mathbf{EFL}^{-}$ by the rule $(\mathsf{F}L)$. We choose an unprovable tree sequent as $\Gamma_{j}^{(h+1)} \stackrel{\mathcal{T}_{j}}{\Rightarrow} \Delta_{j}^{(h+1)}$. Finally we define 
\begin{center}
$\Gamma_{j+1} \stackrel{\mathcal{T}_{j+1}}{\Rightarrow} \Delta_{j+1}$ := $\Gamma_{j}^{(l+1)} \stackrel{\mathcal{T}_{j}}{\Rightarrow} \Delta_{j}^{(l+1)}$. 
\end{center}

\item Let $\varphi_{j}$ be of the form $\mathsf{F}\psi$ and $\alpha_{j}: @_{n_{j}} \mathsf{F}\psi \in \Delta_{j}$. Let $m$ be a fresh agent nominal not occuring in $\Gamma_{j}^{=} \stackrel{\mathcal{T}_{j}}{\Rightarrow} \Delta_{j}$ and define 
\begin{center}
$\Gamma_{j+1} \stackrel{\mathcal{T}_{j+1}}{\Rightarrow} \Delta_{j+1}$ := $@_{n_{j}}\langle \mathsf{F} \rangle m, \Gamma_{j}^{=} \stackrel{\mathcal{T}_{j}}{\Rightarrow} \Delta_{j}, \alpha_{j}: @_{m} \psi$, 
\end{center}
whose unprovability in $\mathsf{T}\mathbf{EFL}^{-}$ is assured by the rule $(\mathsf{F}R)$.

\item Let $\varphi_{j}$ be of the form $\Box\psi$ and $\alpha_{j}: @_{n_{j}} \Box\psi \in \Gamma_{j}^{=}$. Let us enumerate all finite $n_{j}$-children of $\alpha_{j}$ in $\mathcal{T}_{j}$ as $\beta_{1}$, $\ldots$, $\beta_{h}$ and define $\Gamma_{j+1} \stackrel{\mathcal{T}_{j+1}}{\Rightarrow} \Delta_{j+1}$ as $\beta_{1}:@_{n_{j}} \psi, \ldots, \beta_{h}:@_{n_{j}} \psi, \Gamma_{j}^{=} \stackrel{\mathcal{T}_{j}}{\Rightarrow} \Delta_{j}$, which is unprovable in $\mathsf{T}\mathbf{EFL}^{-}$ by the rule $(\Box L)$.

\item Let $\varphi_{j}$ be of the form $\Box\psi$ and $\alpha_{j}: @_{n_{j}} \Box\psi \in \Delta_{j}$. Let $\beta$ be a fresh label not occuring in $\Gamma_{j}^{=} \stackrel{\mathcal{T}_{j}}{\Rightarrow} \Delta_{j}$ such that $\beta$ is an $n_{j}$-child of $\beta$, and define 
\begin{center}
$\Gamma_{j+1} \stackrel{\mathcal{T}_{j+1}}{\Rightarrow} \Delta_{j+1}$ := $\Gamma_{j}^{=} \stackrel{\mathcal{T}_{j} \cup \setof{\beta}}{\Rightarrow} \Delta_{j}, \beta: @_{n_{j}}\psi$, 
\end{center}
whose unprovability in $\mathsf{T}\mathbf{EFL}^{-}$ is assured by the rule $(\Box R)$.

\item Otherwise, $\Gamma_{j+1} \stackrel{\mathcal{T}_{j+1}}{\Rightarrow} \Delta_{j+1}$ is defined as $\Gamma_{j}^{=} \stackrel{\mathcal{T}_{j}}{\Rightarrow} \Delta_{j}$.  

\end{itemize}
\end{description}
We have finished defined a sequence $(\Gamma_{i} \stackrel{\mathcal{T}_{i}}{\Rightarrow} \Delta_{i})_{i \in \omega}$. We define $\Gamma^{+}$ := $\bigcup_{i \in \omega} \Gamma_{i}$, $\mathcal{T}^{+}$ := $\bigcup_{i \in \omega} \mathcal{T}_{i}$ and $\Delta^{+}$ := $\bigcup_{i \in \omega} \Delta_{i}$. Then it is easy to see that $\Gamma^{+} \stackrel{\mathcal{T}^{+}}{\Rightarrow} \Delta^{+}$ is a saturated sequent (we note that the rule $(w\mathsf{lab})$ is needed here). 
\qed
\end{proof}

\begin{lemma}
\label{lem:truth}
Let $\Gamma \stackrel{\mathcal{T}}{\Rightarrow} \Delta$ be a saturated and unprovable tree sequent in $\mathsf{T}\mathbf{EFL}^{-}$. Define the derived model $\mathfrak{M}$ = $(\mathcal{T}, A, (R_{a})_{a \in A}, (\asymp_{\alpha})_{\alpha \in \mathcal{T}},V)$ from $\Gamma \stackrel{\mathcal{T}}{\Rightarrow} \Delta$ by:
\begin{itemize}
\item $A$ $:=$ $\inset{|n|}{\text{ $n$ is an agent nominal }}$, where $|n|$ is an equivalence class of an equivalence relation $\sim$ which is defined as: $n \sim m$ $\iff$ $\alpha:@_{n}m \in \Gamma$ for some $\alpha \in \mathcal{T}$. 
\item $\alpha R_{|n|} \beta$ iff $\beta$ is an $m$-child of $\alpha$ for some $m \in |n|$. 
\item $|n| \asymp_{\alpha} |m|$ iff $\alpha: @_{n} \langle \mathsf{F} \rangle m \in \Gamma$. 
\item $(\alpha, |n|) \in V(m)$ iff $\alpha: @_{n} m \in \Gamma$ $(m \in \mathsf{Nom})$.
\item $(\alpha, |n|) \in V(p)$ iff $\alpha: @_{n} p \in \Gamma$ $(p \in \mathsf{Prop})$.
\end{itemize}
Then, $\mathfrak{M}$ is a model. Moreover, for every labelled formula $\alpha:@_{n}\varphi$, we have:
\begin{itemize}
\item[$(\mathrm{i})$] If $\alpha:@_{n}\varphi \in \Gamma$ then $\mathfrak{M},(\alpha,|n|) \models \varphi$;  
\item[$(\mathrm{ii})$] If $\alpha:@_{n}\varphi \in \Delta$ then $\mathfrak{M},(\alpha,|n|) \not\models \varphi$.
\end{itemize}
\end{lemma}

\begin{proof}
First, let us check that $\mathfrak{M}$ is a model. First of all, note that we can easily verify that $\sim$ is an equivalence relation by the conditions $(\mathbf{ref}_{=})$, $(\mathbf{rep}_{i})$ and $(\mathbf{rigid}_{=})$ of Definition \ref{dfn:saturation}. 
We can also check that if $n \sim m$ then $R_{|n|}$ = $R_{|m|}$ and that if $n \sim n'$ and $m \sim m'$ then $\alpha:@_{n}\langle \mathsf{F}\rangle m \in \Gamma$ iff $\alpha:@_{n'}\langle \mathsf{F}\rangle m' \in \Gamma$. So both of $R_{|n|}$ and $\asymp_{\alpha}$ are well-defined. As for the valuation of propositional variables, when $n \sim m$ holds, the equivalence between $\alpha: @_{n}p \in \Gamma$ and $\alpha: @_{m}p \in \Gamma$ holds by the saturation conditions $(\mathbf{rep}_{1})$ and $(\mathbf{rep}_{2})$. For the valuation for agent nominals $m$, we need to check that $\inset{(\alpha,|n|)}{\alpha: @_{n} m \in \Gamma}$ is $\mathcal{T} \times \setof{|m|}$. But this is clear from the saturation condition $(\mathbf{rigid}_{=})$ and the fact that $\sim$ is an equivalence relation. 

Now we move to check items (i) and (ii) by induction on $\varphi$. We only check the cases where $\varphi$ is of the form: $m$, $\bot$ or $\mathsf{F}\varphi$ or $\Box \varphi$, since the other cases are easy to establish by the corresponding saturation conditions of Definition \ref{dfn:saturation}. 
\begin{itemize}
\item  Let $\varphi$ be of the form $m$. For (i), suppose that $\alpha:@_{n} m \in \Gamma$. This means that $|n|$ = $|m|$. Since $V(m)$ = $\mathcal{T} \times \setof{|m|}$, we have $\mathfrak{M},(\alpha, |n|) \models m$, as desired. For (ii), assume that $\alpha:@_{n} m \in \Delta$ and suppose for contradiction that $\mathfrak{M},(\alpha,|n|) \models m$, i.e., $|n|$ = $|m|$. It follows from $|n|$ = $|m|$ and the saturation condition $(\mathbf{rigid}_{=})$ that $\alpha:@_{n}m \in \Gamma$. This is a contradiction with the unprovability of $\Gamma \stackrel{\mathcal{T}}{\Rightarrow} \Delta$ in $\mathsf{T}\mathbf{EFL}^{-}$. Therefore, we conclude that $\mathfrak{M},(\alpha,|n|) \not\models m$. 

\item  Let $\varphi$ be of the form $\bot$. Since $\Gamma \stackrel{\mathcal{T}}{\Rightarrow} \Delta$ is unprovable in $\mathsf{T}\mathbf{EFL}^{-}$, it is impossible to have $\alpha:@_{n}\bot \in \Gamma$, (i) trivially holds. Since $\mathfrak{M},(\alpha,|n|) \not\models \bot$ always holds, (ii) also holds. 

\item Let $\varphi$ be of the form $\mathsf{F} \varphi$. 
For (i), assume that $\alpha:@_{n} \mathsf{F} \varphi \in \Gamma$. 
We need to show $\mathfrak{M}, (\alpha,|n|)\models \mathsf{F} \varphi$, 
so let us fix any agent nominal $m$ such that $|n| R_{\alpha} |m|$. 
Our goal is to show $\mathfrak{M}, (\alpha,|m|) \models \varphi$. 
From $|n| R_{\alpha} |m|$, we get $\alpha: @_{n}\langle \mathsf{F} \rangle m \in \Gamma$ hence $\alpha: @_{n}\langle \mathsf{F} \rangle m \notin \Delta$ by the unprovability of $\Gamma \stackrel{\mathcal{T}}{\Rightarrow} \Delta$. By the condition $(\mathsf{F}\mathbf{l})$, we obtain $\alpha:@_{m}\varphi \in \Gamma$, which implies our goal by induction hypothesis. 

For (ii), assume that $\alpha:@_{n} \mathsf{F} \varphi \in \Delta$. By the saturation condition $(\mathsf{F}\mathbf{r})$, we have that $\alpha:@_{n}\langle \mathsf{F} \rangle m \in \Gamma$ and $\alpha:@_{m} \varphi \in \Delta$ for some agent nominal $m$.  With the help of induction hypothesis, we have $|n| R_{\alpha} |m|$ and $\mathfrak{M}, (\alpha,|m|) \not\models \varphi$ for some agent nominal $m$. Hence $\mathfrak{M}, (\alpha,|n|) \not\models \mathsf{F} \varphi$, as desired. 

\item Let $\varphi$ be of the form $\Box \varphi$. To show (i), assume that $\alpha:@_{n}\Box \varphi \in \Gamma$. We need to show $\mathfrak{M}, (\alpha,|n|)\models \Box \varphi$, so let us fix any label $\beta$ such that $\alpha R_{|n|} \beta$. Our goal is to show $\mathfrak{M}, (\beta,|n|) \models \varphi$. By $\alpha R_{|n|} \beta$, we can find an agent nominal $m \in |n|$ such that $\beta$ is an $m$-child of $\alpha$. It follows from $m \in |n|$ that $\gamma:@_{n}m \in \Gamma$ for some label $\gamma$. By $\alpha:@_{n}\Box \varphi \in \Gamma$ and $\gamma:@_{n}m \in \Gamma$, the saturation condition ($\mathbf{rep}_{1}$) implies that $\alpha:@_{m}\Box \varphi \in \Gamma$. By the saturation condition  ($\Box \mathbf{l}$) and the fact that $\beta$ is an $m$-child of $\alpha$, we obtain $\beta:@_{m} \varphi \in \Gamma$. By induction hypothesis, $\mathfrak{M}, (\beta,|m|) \models \varphi$ hence we obtain our goal by $|m|$ = $|n|$. This finishes to show (i). 

For (ii), assume that $\alpha:@_{n}\Box \varphi \in \Delta$. By the saturation condition $(\Box \mathbf{r})$, we have that $\beta:@_{n}\varphi \in \Delta$ for some $n$-child $\beta$ of $\alpha$, i.e., $\alpha R_{|n|}\beta$. By induction hypothesis, $\mathfrak{M},(\beta,|n|) \not\models \varphi$. So we conclude that $\mathfrak{M},(\alpha ,|n|) \not\models \Box \varphi$. \qed
\end{itemize}
\end{proof}

\begin{theorem}[Completeness of cut-free $\mathsf{T}\mathbf{EFL}^{-}$]
\label{thm:comp_tree_seq4efl}
If $\mathfrak{M},f \models \Gamma \stackrel{\mathcal{T}}{\Rightarrow} \Delta$ for all models $\mathfrak{M}$ and all assignments $f$, then $\Gamma \stackrel{\mathcal{T}}{\Rightarrow} \Delta$ is provable in $\mathsf{T}\mathbf{EFL}^{-}$.
\end{theorem}

\begin{proof}
Suppose for contradiction that $\Gamma \stackrel{\mathcal{T}}{\Rightarrow} \Delta$ is unprovable in $\mathsf{T}\mathbf{EFL}^{-}$. By Lemma \ref{lem:saturation}, we can extend this tree sequent into a saturated (possibly infinite) tree sequent $\Gamma^{+} \stackrel{\mathcal{T}^{+}}{\Rightarrow} \Delta^{+}$ which is still unprovable in $\mathsf{T}\mathbf{EFL}^{-}$. 
Let $\mathfrak{M}$ be the derived model from $\Gamma^{+} \stackrel{\mathcal{T}^{+}}{\Rightarrow} \Delta^{+}$. Let us define $f: \mathcal{T} \to \mathcal{T}$ as the identity mapping. Then it follows from Lemma \ref{lem:truth} that $\mathfrak{M},f\not\models \Gamma \Rightarrow \Delta$, as required. \qed
\end{proof}

By Theorems \ref{thm:sound_tree_seq4efl} and \ref{thm:comp_tree_seq4efl}, the cut elimination theorem of $\mathsf{T}\mathbf{EFL}$ follows. 

\begin{corollary}
\label{cor:sum_tree}
The following are all equivalent: 

\begin{enumerate}
\item $\mathfrak{M},f \models \Gamma \stackrel{\mathcal{T}}{\Rightarrow} \Delta$ for all models $\mathfrak{M}$ and all assignments $f$. 
\item $\Gamma \stackrel{\mathcal{T}}{\Rightarrow} \Delta$ is provable in $\mathsf{T}\mathbf{EFL}^{-}$.
\item $\Gamma \stackrel{\mathcal{T}}{\Rightarrow} \Delta$ is provable in $\mathsf{T}\mathbf{EFL}$.
\end{enumerate}
Therefore, $\mathsf{T}\mathbf{EFL}$ enjoys the cut-elimination theorem. 
\end{corollary}

\section{Hilbert System of Epistemic Logic of Friendship}
\label{sec:hil_helf}

This section provides a Hilbert system of the epistemic logic of friendship by ``translating'' a tree sequent into a formula in $\mathcal{L}$. 
First of all, let us introduce the notion of {\em necessity form}, originally proposed in~\cite{Goldblatt1982} by Goldblatt and used also in~\cite{BT1998,GPT1987}. 
Necessity forms are employed to formulate an inference rule of our Hilbert system. 

\begin{definition}[Necessity form]
Fix an arbitrary symbol $\#$ not occurring in the syntax $\mathcal{L}$. A {\em necessity form} is defined inductively as follows: $(\mathrm{i})$ $\#$ is a necessity form; $(\mathrm{ii})$ If $L$ is a necessity form and $\varphi$ is a formula, then $\varphi \to L$ is also a necessity form; $(\mathrm{iii})$ If $L$ is a necessity form and $n$ is an agent nominal, then $@_{n}\Box L$ is also a necessity form. Given a necessity form $L(\#)$ and a formula $\varphi$ of $\mathcal{L}$, we use $L(\varphi)$ to denote the formula obtained by replacing the unique occurrence of $\#$ in $L$ by the formula $\varphi$. 
\end{definition}

When $L(\#)$ is a necessity form of $\psi_{0} \to @_{n}\Box (\psi_{1} \to @_{m}\Box (\psi_{2} \to \#))$, then $L(\varphi)$ is $\psi_{0} \to @_{n}\Box (\psi_{1} \to @_{m}\Box (\psi_{2} \to \varphi))$. Intuitively, this notion allows us to capture the unique path from a label in a tree of a tree sequent to the root label of the tree. 

\begin{table}[htbp]
\hrule
\vspace{1pt}
\caption{Hilbert System $\mathsf{H}\mathbf{EFL}$}
\label{table:hilbert}
\begin{center}
\begin{tabular}{llll}
(\texttt{Taut}) & all propositional tautologies & (\texttt{MP})& From $\varphi$ and $\varphi \to \psi$, infer $\psi$ \\
(\texttt{K}$_{\Box}$)& $\Box (\varphi \to \psi) \to (\Box \varphi \to \Box \psi)$ & (\texttt{Nec}$_{\Box}$)& From $\varphi$, infer $\Box \varphi$ \\
(\texttt{K}$_{\mathsf{F}}$)& $\mathsf{F} (\varphi \to \psi) \to (\mathsf{F} \varphi \to \mathsf{F} \psi)$ & (\texttt{Nec}$_{\mathsf{F}}$)& From $\varphi$, infer $\mathsf{F} \varphi$ \\
(\texttt{K}$_@$) &  $@_{n} (\varphi \to \psi) \to (@_{n} \varphi \to @_{n} \psi)$ & (\texttt{Nec}$_{@}$)& From $\varphi$, infer $@_{n} \varphi$ \\
(\texttt{Ref}) & $@_{n}n$  & (\texttt{Selfdual}) & $\neg @_{n} \varphi \leftrightarrow @_{n}\neg \varphi$ \\
(\texttt{Elim}) & $@_{n}\varphi \to (n \to \varphi)$ & (\texttt{Agree}) & $@_{n}@_{m}\varphi \to @_{m}\varphi$  \\
(\texttt{Back}) & $@_{n}\varphi \to \mathsf{F}@_{n}\varphi$ & (\texttt{DCom}$@\Box$) &  $@_{n} \Box @_{n} \varphi \leftrightarrow  @_{n} \Box \varphi$ \\
(\texttt{Rigid}$_{=}$) & $@_{n}m \to \Box @_{n}m$ & (\texttt{Rigid}$_{\neq}$) & $\neg @_{n}m \to \Box \neg @_{n}m$ \\
(\texttt{Name}) &\multicolumn{3}{l}{ From $n \to \varphi$, infer $\varphi$, 
where $n$ is fresh in $\varphi$. }\\
($L$(\texttt{BG}))& \multicolumn{3}{l}{ From $L(@_{n} \langle \mathsf{F} \rangle m \to @_{m}\varphi)$, infer $L(@_{n} \mathsf{F} \varphi)$, where $m$ is fresh in $L(@_{n} \mathsf{F} \varphi)$.}\\
\end{tabular}
\end{center}
\hrule
\end{table}

Table \ref{table:hilbert} presents our Hilbert system $\mathsf{H}\mathbf{EFL}$. The underlying idea of the system is the following. On the top of the propositional part ($\texttt{Taut}$ and $\texttt{MP}$), we combine the axiomatization of modal logic $\mathbf{K}$ for the modal operator $\Box$ and the axiomatization of a basic hybrid logic $\mathbf{K}_{\mathcal{H}(@)}$ (see~\cite{BC2006,BdRV2001}) for the modal operator $\mathsf{F}$, with some modification (we need to modify $\mathbf{BG}$, the rule of {\em bounded generalization}, with the help of necessity forms), and then we add three interaction axioms: (\texttt{Rigid}$_{=}$), (\texttt{Rigid}$_{\neq}$), and (\texttt{DCom}$@\Box$). We note that the axiom (\texttt{DCom}$@\Box$) is also used for axiomatizing the {\em dependent product} of two hybrid logics in~\cite{Sano2010a}. Let us define the notion of provability in $\mathsf{H}\mathbf{EFL}$ in as usual. We write $\vdash_{\mathsf{H}\mathbf{EFL}} \varphi$ to means that $\varphi$ is provable in $\mathsf{H}\mathbf{EFL}$. 
\footnote{
By (\texttt{K})-rules and (\texttt{Nec})-rules for operators $\Box$, $\mathsf{F}$ and $@_{n}$, the replacement of equivalence holds in $\mathsf{H}\mathbf{EFL}$. }
\footnote{Given a set $\Gamma \cup \setof{\varphi}$ of formulas, we say that $\varphi$ is {\em deducible} in $\mathsf{H}\mathbf{EFL}$ from $\Gamma$ if there exist finite formulas $\psi_{1}$, $\ldots$, $\psi_{n} \in \Gamma$ such that $(\psi_{1}\land \cdots \land \psi_{n}) \to \varphi$ is provable in $\mathsf{H}\mathbf{EFL}$. Then it is easy to see that the deduction theorem holds in $\mathsf{H}\mathbf{EFL}$. }

\begin{proposition}
\label{prop:hil_unisub}
Uniform substitutions are length-preserving admissible in $\mathsf{H}\mathbf{EFL}$, i.e., if $\sigma$ is a uniform substitution and $\varphi$ has a derivation in $\mathsf{H}\mathbf{EFL}$ whose length is at most $n$, then $\varphi\sigma$ has a derivation in $\mathsf{H}\mathbf{EFL}$ whose length is at most $n$. 
\end{proposition}

\begin{proposition}
\label{prop:hilsys}
All the following are provable in $\mathsf{H}\mathbf{EFL}$. 
\begin{enumerate}
\item $@_{m}@_{n}\varphi \leftrightarrow @_{n}\varphi$. 
\item $n \to (@_{n} \varphi \leftrightarrow \varphi)$. 
\item $@_{n}m \to (@_{n} \varphi \leftrightarrow @_{m}\varphi)$. 
\item $@_{n}m \leftrightarrow @_{m}n$. 
\item $@_{n}(\varphi \to \psi) \leftrightarrow (@_{n}\varphi \to @_{n}\psi)$. 
\item $@_{n}m \to (\varphi[n/k] \leftrightarrow \varphi[m/k])$. 
\end{enumerate}
\end{proposition}

\begin{proof}
For the provability of item 1, it suffices to show the right-to-left direction, which is shown by (\texttt{Agree}) and (\texttt{Selfdual}). For the provability of item 2, it suffices to show $n \to (\varphi \to @_{n} \varphi)$, whose provability is shown by the contraposition of $(\texttt{Elim})$ and $(\texttt{Selfdual})$. Then items 3 to 5 are proved similarly as given in~\cite[p.293, Lemma 2]{BC2006}. Finally, item 6 is proved by induction on $\varphi$. Here we show the case where $\varphi$ is of the form $l \in \mathsf{Nom}$, $\Box \psi$ and $@_{l}\psi$. First, we consider the case where $\varphi$ is of the form $l \in \mathsf{Nom}$. When $l \not\equiv k$, there is nothing to prove, so we focus on the case where $l \equiv k$. It suffices to show that $\vdash_{\mathsf{H}\mathbf{EFL}} @_{n}m \to (n \leftrightarrow m)$, but this is clear from items 2 and 4.  Second, we move to the case where $\varphi$ is of the form $\Box \psi$. By induction hypothesis, we obtain $\vdash_{\mathsf{H}\mathbf{EFL}} @_{n}m \to (\psi[n/k] \leftrightarrow \psi[m/k])$. By $(\texttt{K}_{\Box})$ and $(\texttt{Nec}_{\Box})$, we get $\vdash_{\mathsf{H}\mathbf{EFL}} \Box @_{n}m \to (\Box (\psi[n/k]) \leftrightarrow \Box (\psi[m/k]))$. It follows from the axiom $(\texttt{rigid}_{=})$ that $\vdash_{\mathsf{H}\mathbf{EFL}} @_{n}m \to((\Box \psi)[n/k] \leftrightarrow (\Box \psi)[m/k]))$, as desired.  
Third, we deal with the case where $\varphi$ is of the form $@_{l} \psi$. 
When $l \equiv k$, we show that $\vdash_{\mathsf{H}\mathbf{EFL}} @_{n}m \to (@_{l} (\psi[n/k]) \leftrightarrow @_{l} (\psi[m/k]))$. This is easily obtained by induction hypothesis, $(\texttt{Nec}_{@})$ and items 1 and 5. When $l \not\equiv k$, it suffices to prove that $\vdash_{\mathsf{H}\mathbf{EFL}} @_{n}m \to (@_{n} (\psi[n/k]) \leftrightarrow @_{m} (\psi[m/k]))$. By induction hypothesis, we have $\vdash_{\mathsf{H}\mathbf{EFL}} @_{n}m \to ((\psi[n/k]) \leftrightarrow (\psi[m/k]))$. By $(\texttt{Nec}_{@})$, we have 
\[
\vdash_{\mathsf{H}\mathbf{EFL}} @_{n}@_{n}m \to @_{n}((\psi[n/k]) \leftrightarrow (\psi[m/k])).
\] By items 2 and 5, 
\[
\vdash_{\mathsf{H}\mathbf{EFL}} @_{n}m \to (@_{n}(\psi[n/k]) \leftrightarrow @_{n}(\psi[m/k])).
\] 
By items 3 and 4, 
\[
\vdash_{\mathsf{H}\mathbf{EFL}} @_{n}m \to (@_{n}(\psi[m/k]) \leftrightarrow @_{m}(\psi[m/k])).
\]
This allows us to conclude $\vdash_{\mathsf{H}\mathbf{EFL}} @_{n}m \to (@_{n} (\psi[n/k]) \leftrightarrow @_{m} (\psi[m/k]))$. \qed
\end{proof}

The following translation is a key to specify our Hilbert system $\mathsf{H}\mathbf{EFL}$. 
\begin{definition}[Formulaic translation]
Given a set $\Theta$ of labelled formulas and a label $\alpha$, we define $\Theta_{\alpha}$ $:=$ $\inset{\varphi}{\alpha:\varphi \in \Theta}$. Let $\Gamma \stackrel{\mathcal{T}}{\Rightarrow} \Delta$ be a tree sequent. Then the {\em formulaic translation} of the sequent at $\alpha$ is defined inductively as:
\[
\left[\!\!\left[ \Gamma \stackrel{\mathcal{T}}{\Rightarrow} \Delta \right]\!\!\right]_{\alpha} := \bigwedge \Gamma_{\alpha} \to \bigvee\left(  \Delta_{\alpha}, @_{n_{1}}\Box \left[\!\!\left[ \Gamma \stackrel{\mathcal{T}}{\Rightarrow} \Delta \right]\!\!\right]_{\beta_{1}}, \ldots, @_{n_{k}}\Box \left[\!\!\left[ \Gamma \stackrel{\mathcal{T}}{\Rightarrow} \Delta \right]\!\!\right]_{\beta_{k}}
\right),
\]
where $\beta_{i}$ is an $n_{i}$-child of $\alpha$, $\beta_{i}$s enumerate all children of $\alpha$, $\bigwedge \emptyset$ $:=$ $\top$, and $\bigvee \emptyset$ $:=$ $\bot$.  \end{definition}

The formulaic translation of a tree sequent of Fig.~\ref{fig:tree_seq} of Section \ref{sec:tree_seq} at the root $0$ is 
\begin{center}
$@_{n}\varphi \to (@_{m}\psi \lor @_{n}\Box(\top \to @_{k}\theta) \lor @_{k}\Box(@_{m} \rho \to \bot))$. 
\end{center}

\begin{theorem}
\label{thm:comp_hil4efl}
If a tree sequent $\Gamma \stackrel{\mathcal{T}}{\Rightarrow} \Delta$ is provable in $\mathsf{T}\mathbf{EFL}$ then the formulaic translation $[\![ \Gamma \stackrel{\mathcal{T}}{\Rightarrow} \Delta ]\!]_{i}$ is provable in $\mathsf{H}\mathbf{EFL}$, where a natural number $i$ is the root of $\mathcal{T}$. 
\end{theorem}

\begin{proof}
By induction on height $n$ of a derivation of $\Gamma \stackrel{\mathcal{T}}{\Rightarrow} \Delta$ in $\mathsf{T}\mathbf{EFL}$, where $i$ is the root of the tree $\mathcal{T}$. We skip the base case where $n$ = $0$. Let $n>0$. It is remarked that, when the sequent is obtained by $(\mathsf{rep}_{l})$, $(\mathsf{ref}_{=})$, $(@ L)$, or $(@ R)$, respectively, the translation of the sequent at the root is provable by Proposition \ref{prop:hilsys} (6), the axiom (\texttt{Ref}), (\texttt{Agree}), or Proposition \ref{prop:hilsys} (1), respectively. Here we focus on the cases where $\Gamma \stackrel{\mathcal{T}}{\Rightarrow} \Delta$ is obtained by $(\Box L)$, $(\mathsf{F}R)$ or $(\mathsf{rigid}_{=})$, since these are the cases where we need to be careful and the other cases are easy to establish. 
\begin{itemize}
\item[($\Box L$)] Suppose that $\alpha:@_{n}\Box \varphi, \Gamma' \stackrel{\mathcal{T}}{\Rightarrow} \Delta$ is obtained by ($\Box L$) from $\beta:@_{n}\varphi, \Gamma' \stackrel{\mathcal{T}}{\Rightarrow} \Delta$, where $\beta \in \mathcal{T}$ is an $n$-child of $\alpha$. By induction hypothesis, we obtain $\vdash_{\mathsf{H}\mathbf{EFL}}   \left[\!\!\left[ \beta:@_{n}\varphi, \Gamma' \stackrel{\mathcal{T}}{\Rightarrow} \Delta \right]\!\!\right]_{i}$. We show that $\vdash_{\mathsf{H}\mathbf{EFL}}   \left[\!\!\left[ \alpha:@_{n}\Box \varphi, \Gamma' \stackrel{\mathcal{T}}{\Rightarrow} \Delta \right]\!\!\right]_{i}$. 
Let $(\alpha_{0}, \alpha_{1}, \ldots, \alpha_{l})$ be the unique path from $\alpha$ ($\equiv$ $\alpha_{l}$) to the root $i$ ($\equiv$ $\alpha_{0}$) of tree $\mathcal{T}$. 
By induction on $0 \leqslant h \leqslant l$, we show that 
\[
\vdash_{\mathsf{H}\mathbf{EFL}} 
\left[\!\!\left[ \beta:@_{n}\varphi, \Gamma' \stackrel{\mathcal{T}}{\Rightarrow} \Delta \right]\!\!\right]_{\alpha_{l-h}}
\to 
 \left[\!\!\left[ \alpha:@_{n}\Box \varphi, \Gamma' \stackrel{\mathcal{T}}{\Rightarrow} \Delta \right]\!\!\right]_{\alpha_{l-h}}. 
\]

Let $h$ = $0$ and so $\alpha_{l-h}$ = $\alpha$. It suffices to show that a formula of the form
\[
\left( \gamma_{1} \to (\delta \lor @_{n}\Box ( (\gamma_{2} \land @_{n}\varphi) \to \psi_{2}) \right) \to 
\left( (@_{n}\Box \varphi  \land \gamma_{1} ) \to (\delta \lor @_{n}\Box ( \gamma_{2} \to \psi_{2}))  \right). 
\]
is provable in ${\mathsf{H}\mathbf{EFL}}$. This reduces to the provability of 
\[
@_{n}\Box \varphi  \land @_{n}\Box ( (\gamma_{2} \land @_{n}\varphi) \to \psi_{2})) \to @_{n}\Box ( \gamma_{2} \to \psi_{2}))
\]
in ${\mathsf{H}\mathbf{EFL}}$. This holds by the axiom $(\texttt{Dcom}\Box @)$ $@_{n}\Box @_{n} \varphi \leftrightarrow  @_{n}\Box \varphi$. 

Let $h>0$. But this case is shown with the help of $(\mathsf{Nec}_{\Box})$ and $(\mathsf{Nec}_{@})$. This completes our induction on $h$. So we conclude $\vdash_{\mathsf{H}\mathbf{EFL}}   \left[\!\!\left[ \alpha:@_{n}\Box \varphi, \Gamma' \stackrel{\mathcal{T}}{\Rightarrow} \Delta \right]\!\!\right]_{i}$. 

\item[($\mathsf{F}R$)] Suppose that $\Gamma \stackrel{\mathcal{T}}{\Rightarrow} \Delta', \alpha:@_{n}\mathsf{F}\varphi$ is obtained by ($\mathsf{F}R$) from  $\alpha: @_{n}\langle \mathsf{F}\rangle m, \Gamma \stackrel{\mathcal{T}}{\Rightarrow} \Delta', \alpha:@_{m} \varphi$ where $m$ is fresh in the conclusion. By induction hypothesis, we have $\vdash_{\mathsf{H}\mathbf{EFL}} \left[\!\!\left[ \alpha: @_{n}\langle \mathsf{F}\rangle m, \Gamma \stackrel{\mathcal{T}}{\Rightarrow} \Delta', \alpha:@_{m} \varphi \right]\!\!\right]_{i}$, which is equivalent to $\vdash_{\mathsf{H}\mathbf{EFL}} L(@_{n}\langle \mathsf{F}\rangle m \to @_{m} \varphi)$ for some necessitation form $L$. Fix such necessitation form $L$. By the inference rule $L(\texttt{BG})$ of $\mathsf{H}\mathbf{EFL}$, we can obtain $\vdash_{\mathsf{H}\mathbf{EFL}} L(@_{n} \mathsf{F} \varphi)$, which is equivalent to 
$\vdash_{\mathsf{H}\mathbf{EFL}}  \left[\!\!\left[  \Gamma \stackrel{\mathcal{T}}{\Rightarrow} \Delta', \alpha:@_{n}\mathsf{F}\varphi  \right]\!\!\right]_{i}$. 
\item[$(\mathsf{rigid}_{=})$] 
Let us suppose that $\alpha: @_{n}m, \Gamma \stackrel{\mathcal{T}}{\Rightarrow} \Delta$ is obtained by $(\mathsf{rigid}_{=})$ from $\beta: @_{n}m, \Gamma  \stackrel{\mathcal{T}}{\Rightarrow} \Delta$. By induction hypothesis, we obtain $\vdash_{\mathsf{H}\mathbf{EFL}}  \left[\!\!\left[  \beta: @_{n}m, \Gamma  \stackrel{\mathcal{T}}{\Rightarrow} \Delta \right]\!\!\right]_{i}$. Our goal is to show that 
$\vdash_{\mathsf{H}\mathbf{EFL}}  \left[\!\!\left[  \alpha: @_{n}m, \Gamma  \stackrel{\mathcal{T}}{\Rightarrow} \Delta \right]\!\!\right]_{i}$. 
It suffices to show the following two cases: (i) $\beta$ is a $k$-child of $\alpha$ or (ii) $\alpha$ is a $k$-child of $\beta$.  We note that we will use the axioms $(\texttt{Rigid}_{=})$ in (i) and $(\texttt{Rigid}_{\neq})$ in (ii). First, we deal with the case (i). Let $(\alpha_{0}, \alpha_{1}, \ldots, \alpha_{l})$ be the unique path from $\alpha$ ($\equiv$ $\alpha_{l}$) to the root $i$ ($\equiv$ $\alpha_{0}$) of tree $\mathcal{T}$. Recall that we assume that $\beta$ is a $k$-child of $\alpha$. By induction on $0 \leqslant h \leqslant l$, we show that $\vdash_{\mathsf{H}\mathbf{EFL}}  \left[\!\!\left[  \beta: @_{n}m, \Gamma  \stackrel{\mathcal{T}}{\Rightarrow} \Delta \right]\!\!\right]_{\alpha_{l-h}} \to \left[\!\!\left[  \alpha: @_{n}m, \Gamma  \stackrel{\mathcal{T}}{\Rightarrow} \Delta \right]\!\!\right]_{\alpha_{l-h}}$. 
Let $h$ = $0$ and so $\alpha_{l-h}$ = $\alpha$. It suffices to show that a formula of the form: 
\[
(\gamma_{\alpha} \to (\delta_{\alpha} \lor @_{k}\Box ((@_{n}m \land \gamma_{\beta}) \to \delta_{\beta} )) ) 
\to ((\gamma_{\alpha}\land @_{n}m )\to (\delta_{\alpha} \lor @_{k}\Box (\gamma_{\beta} \to \delta_{\beta} )) ) 
\]
is provable in ${\mathsf{H}\mathbf{EFL}}$. For this, it suffices to show $\vdash_{\mathsf{H}\mathbf{EFL}} @_{n}m \to @_{k}\Box @_{n}m$, which holds by $(\texttt{Rigid}_{=})$, the distribution of $@$ over the implication and Proposition \ref{prop:hilsys} (1). Let $h>0$. But this case is shown with the help of $(\mathsf{Nec}_{\Box})$ and $(\mathsf{Nec}_{@})$. This completes our induction on $h$. So we conclude $\vdash_{\mathsf{H}\mathbf{EFL}}   \left[\!\!\left[ \alpha:@_{n}m, \Gamma \stackrel{\mathcal{T}}{\Rightarrow} \Delta \right]\!\!\right]_{i}$. 
Second, we move to the case (ii). Let $(\beta_{0}, \beta_{1}, \ldots, \beta_{l})$ be the unique path from $\beta$ ($\equiv$ $\beta_{l}$) to the root $i$ ($\equiv$ $\beta_{0}$) of tree $\mathcal{T}$. Note that we assume that $\alpha$ is a $k$-child of $\beta$. By induction on $0 \leqslant h \leqslant l$, we show that $\vdash_{\mathsf{H}\mathbf{EFL}}  \left[\!\!\left[  \beta: @_{n}m, \Gamma  \stackrel{\mathcal{T}}{\Rightarrow} \Delta \right]\!\!\right]_{\beta_{l-h}} \to \left[\!\!\left[  \alpha: @_{n}m, \Gamma  \stackrel{\mathcal{T}}{\Rightarrow} \Delta \right]\!\!\right]_{\beta_{l-h}}$. Let $h$ = $0$ and so $\beta_{l-h}$ = $\beta$. It suffices to show that a formula of the form: 
\[
((\gamma_{\beta}\land @_{n}m )\to (\delta_{\beta} \lor @_{k}\Box (\gamma_{\alpha} \to \delta_{\alpha} )) ) \to
(\gamma_{\beta} \to (\delta_{\beta} \lor @_{k}\Box ((@_{n}m \land \gamma_{\alpha}) \to \delta_{\alpha} )) ) 
\]
is provable in ${\mathsf{H}\mathbf{EFL}}$. For this, it suffices to show $\vdash_{\mathsf{H}\mathbf{EFL}} \neg @_{n}m \to @_{k}\Box \neg @_{n}m$, which holds by $(\texttt{Rigid}_{\neq})$, $(\texttt{Selfdual})$ and Proposition \ref{prop:hilsys} (1). Let $h>0$. But this case is shown with the help of $(\mathsf{Nec}_{\Box})$ and $(\mathsf{Nec}_{@})$. This completes our induction on $h$. So we conclude $\vdash_{\mathsf{H}\mathbf{EFL}}   \left[\!\!\left[ \alpha:@_{n}m, \Gamma \stackrel{\mathcal{T}}{\Rightarrow} \Delta \right]\!\!\right]_{i}$. 
\qed
\end{itemize}
\end{proof}

In what follows in this section, we prove the soundness of $\mathsf{H}\mathbf{EFL}$ for the tree sequent calculus $\mathsf{T}\mathbf{EFL}$ with the cut rule. The cut rule is necessary to prove the following. 

\begin{lemma}
\label{lem:inv}
The rules $(\to R)$, $(\Box R)$, $(@ R)$, and $(@ L)$ are invertible, i.e., if the lower sequent is provable in $\mathsf{T}\mathbf{EFL}$ then the upper sequent is also provable in $\mathsf{T}\mathbf{EFL}$. 
\end{lemma}

\begin{proof}
We only prove the invertibility of $(\to R)$ and $(\Box R)$. First we deal with $(\to R)$. 
Suppose that $\Gamma \stackrel{\mathcal{T}}{\Rightarrow} \Delta, \alpha:@_{n}(\varphi \to \psi)$ is provable in $\mathsf{T}\mathbf{EFL}$. This is shown as follows:
\[
\infer[(Cut)]{\alpha:@_{n}\varphi, \Gamma \stackrel{\mathcal{T}}{\Rightarrow} \Delta, \alpha:@_{n}\psi }{
\Gamma \stackrel{\mathcal{T}}{\Rightarrow} \Delta, \alpha:@_{n}(\varphi \to \psi)
&
\alpha:@_{n}(\varphi \to \psi),  \alpha:@_{n}\varphi \stackrel{\mathcal{T}}{\Rightarrow} \alpha:@_{n}\psi
},
\]
where the rightmost tree sequent is provable in $\mathsf{T}\mathbf{EFL}$ by $(\to L)$. 
Second we move to $(\Box R)$. Suppose that $\Gamma \stackrel{\mathcal{T}}{\Rightarrow} \Delta, \alpha : @_{n}\Box \varphi$ is provable in $\mathsf{T}\mathbf{EFL}$. Then the provability of the upper sequent of $(\Box R)$ is established as follows: 
\[
\infer[(Cut)]{
\Gamma \stackrel{\mathcal{T}\cup \setof{\alpha \cdot_{n} i}}{\Rightarrow} \Delta, \alpha\cdot_{n} i: @_{n}\varphi}
{
\infer[(w\mathsf{lab})]{\Gamma \stackrel{\mathcal{T}\cup \setof{\alpha \cdot_{n} i}}{\Rightarrow} \Delta, \alpha: @_{n}\Box \varphi}{\Gamma \stackrel{\mathcal{T}}{\Rightarrow} \Delta, \alpha: @_{n}\Box \varphi}
&
\infer[(L\Box)]{\alpha: @_{n}\Box \varphi, \Gamma \stackrel{\mathcal{T}\cup \setof{\alpha \cdot_{n} i}}{\Rightarrow} \Delta, \alpha\cdot_{n} i: @_{n}\varphi}{
\infer[(\mathsf{id})]{
\alpha\cdot_{n} i: @_{n} \varphi, \Gamma \stackrel{\mathcal{T}\cup \setof{\alpha \cdot_{n} i}}{\Rightarrow} \Delta, \alpha\cdot_{n} i: @_{n}\varphi}{}
}
}.
\]
\qed
\end{proof}

\begin{theorem}
\label{thm:sound_hil4efl}
If $\varphi$ is provable in $\mathsf{H}\mathbf{EFL}$, then $\stackrel{\mathcal{T}}{\Rightarrow} \alpha: @_{n}\varphi$ is provable in $\mathsf{T}\mathbf{EFL}$ for all trees $\mathcal{T}$, $\alpha \in \mathcal{T}$ and nominals $n$. 
\end{theorem}

\begin{proof}
Suppose that there is a derivation $(\varphi_{0}, \ldots, \varphi_{h})$ of $\varphi$ in $\mathsf{H}\mathbf{EFL}$. By induction on $0 \leqslant j \leqslant h$, we show that $\stackrel{\mathcal{T}}{\Rightarrow} \alpha: @_{n}\varphi_{j}$ is provable in $\mathsf{T}\mathbf{EFL}$ for all nominals $n$. We demonstrate some cases. Let us start with (\texttt{Rigid}$_{=}$), which is shown as follows.
\[
\infer[(\to R)]{\stackrel{\mathcal{T}}{\Rightarrow} \alpha:@_{k}(@_{n}m \to \Box @_{n}m) }{
\infer[(@ L)]{\alpha:@_{k} @_{n}m \stackrel{\mathcal{T}}{\Rightarrow} \alpha:@_{k} \Box @_{n}m}{
\infer[(\Box R)]{
\alpha:@_{n}m \stackrel{\mathcal{T}}{\Rightarrow} \alpha:@_{k} \Box @_{n}m
}{
\infer[( @R)]{
\alpha:@_{n}m \stackrel{\mathcal{T}\cup \setof{\alpha \cdot_{k} i}}{\Rightarrow} \alpha \cdot_{k} i:@_{k} @_{n}m
}{
\infer[(\mathsf{\mathsf{rigid}_{=}})]{\alpha:@_{n}m \stackrel{\mathcal{T}\cup \setof{\alpha \cdot_{k} i}}{\Rightarrow} \alpha \cdot_{k} i: @_{n}m
}{
\infer[(\mathsf{id})]{
\alpha \cdot_{k} i:@_{n}m \stackrel{\mathcal{T}\cup \setof{\alpha \cdot_{k} i}}{\Rightarrow} \alpha \cdot_{k} i: @_{n}m}{}
}
}
}
}
}
\]
For (\texttt{Rigid}$_{\neq}$), the following derivation is enough for our goal: 
\[
\infer[(\to R)]{\stackrel{\mathcal{T}}{\Rightarrow} \alpha:@_{k}(\neg @_{n}m \to \Box \neg @_{n}m)}{
\infer[(\Box R)]{\alpha:@_{k}\neg @_{n}m \stackrel{\mathcal{T}}{\Rightarrow} \alpha:@_{k} \Box \neg @_{n}m }{
\infer[(\neg R)]{\alpha:@_{k}\neg @_{n}m \stackrel{ \mathcal{T}\cup \setof{\alpha \cdot_{k} i} }{\Rightarrow} \alpha \cdot_{k} i: @_{k} \neg @_{n}m}{
\infer[(\neg L)]{\alpha \cdot_{k} i: @_{k} @_{n}m, \alpha:@_{k}\neg @_{n}m \stackrel{ \mathcal{T}\cup \setof{\alpha \cdot_{k} i} }{\Rightarrow} }{
\infer[(@R)]{\alpha \cdot_{k} i: @_{k} @_{n}m \stackrel{\mathcal{T} \cup \setof{\alpha \cdot_{k} i} }{\Rightarrow} \alpha:@_{k}@_{n}m }{
\infer[(@L)]{\alpha \cdot_{k} i: @_{k} @_{n}m \stackrel{\mathcal{T} \cup \setof{\alpha \cdot_{k} i} }{\Rightarrow} \alpha:@_{n}m}{
\infer[(\mathsf{rigid}_{=})]{\alpha \cdot_{k} i: @_{n}m \stackrel{ \mathcal{T} \cup \setof{\alpha \cdot_{k} i} }{\Rightarrow} \alpha:@_{n}m}{
\infer[(\mathsf{id})]{
\alpha: @_{n}m \stackrel{\mathcal{T} \cup \setof{\alpha \cdot_{k} i} }{\Rightarrow} \alpha:@_{n}m}
{}
}
}
}
}
}
}
}.
\]
Now we move to (\texttt{DCom}$@\Box$). We show the right-to-left direction alone, since the converse direction is shown similarly. Let us see the derivation below, from which we can obtain the provability of $\stackrel{\mathcal{T}}{\Rightarrow} \alpha:@_{m}(@_{n} \Box @_{n} p \to  @_{n} \Box p)$ in $\mathsf{T}\mathbf{EFL}$:
\[
\infer[(\Box R)]{
\alpha:@_{n} \Box @_{n} p  \stackrel{\mathcal{T}}{\Rightarrow} \alpha:@_{n} \Box p
}{
\infer[(\Box L)]{
\alpha:@_{n} \Box @_{n} p  \stackrel{\mathcal{T}\cup\setof{\alpha \cdot_{n} i}}{\Rightarrow} \alpha:\cdot_{n} i: @_{n} p
}{
\infer[(@L)]{
\alpha \cdot_{n} i :@_{n} @_{n} p  \stackrel{\mathcal{T}\cup\setof{\alpha \cdot_{n} i}}{\Rightarrow} \alpha \cdot_{n} i: @_{n} p
}{
\infer[(\mathsf{id})]{
\alpha \cdot_{n} i :@_{n} p  \stackrel{\mathcal{T}\cup\setof{\alpha \cdot_{n} i}}{\Rightarrow} \alpha \cdot_{n} i: @_{n} p}{}
}
}
}
\]
Now we deal with some inference rules below.
\begin{itemize}
\item[($\texttt{Name}$)] 
Let $\varphi_{j}$ $\equiv$ $n \to \psi$ be obtained by $(\texttt{Name})$. Fix any finite tree $\mathcal{T}$, $\alpha \in \mathcal{T}$ and nominal $k$. Let $m \not\equiv k$ be a fresh nominal in $\mathcal{T}$ and $\psi$. Note that $m$ is also fresh in $\alpha \in \mathcal{T}$. By Proposition \ref{prop:hil_unisub}, $m \to \psi$ has a derivation whose length is {\em at most} $j$. By induction hypothesis, $\stackrel{\mathcal{T}}{\Rightarrow} \alpha: @_{k}(m \to \psi)$ is provable in $\mathsf{T}\mathbf{EFL}$. By admissibility of uniform substitution $[k/m]$ in $\mathsf{T}\mathbf{EFL}$ (by Proposition \ref{prop:admissible}), $\stackrel{\mathcal{T}}{\Rightarrow} \alpha: @_{k}(k \to \psi)$ is provable in $\mathsf{T}\mathbf{EFL}$. By Lemma \ref{lem:inv}, we obtain the provability of $\alpha: @_{k}k \stackrel{\mathcal{T}}{\Rightarrow} \alpha: @_{k}\psi$ in $\mathsf{T}\mathbf{EFL}$. By $(\mathsf{ref}_{=})$, we conclude that $\stackrel{\mathcal{T}}{\Rightarrow} \alpha: @_{k}\psi$ is provable in $\mathsf{T}\mathbf{EFL}$.

\item[($L$(\texttt{BG}))] 
Let $\varphi_{j}$ $\equiv$ $\Box \psi$ be obtained by $(L(\texttt{BG}))$. Fix any finite tree $\mathcal{T}$, $\alpha \in \mathcal{T}$ and nominal $k$. 
By induction hypothesis, $\stackrel{\mathcal{T}}{\Rightarrow} \alpha: @_{k} L(@_{n} \langle \mathsf{F} \rangle m \to @_{m}\varphi)$ is provable in $\mathsf{T}\mathbf{EFL}$, where we can assume that $m$ satisfies the freshness condition by Proposition \ref{prop:hil_unisub}. By applying Lemma \ref{lem:inv} (i.e., the invertibility of the right rules) repeatedly to the consequent of a resulting tree sequent,  we obtain the provability of a tree sequent of the form 
$\Gamma, \beta: @_{n} \langle \mathsf{F} \rangle m \stackrel{\mathcal{T}'}{\Rightarrow}  \Delta, \beta: @_{m}\varphi$. Then we apply the right rules in a converse direction of our repeated application of Lemma \ref{lem:inv} to conclude that $\stackrel{\mathcal{T}}{\Rightarrow} \alpha: @_{k} L(@_{n}  \mathsf{F} \varphi)$ is provable in $\mathsf{T}\mathbf{EFL}$. To illustrate this argument, let $L$ $\equiv$ $@_{n}\Box (\psi \to \#)$. 
By induction hypothesis, $\stackrel{\mathcal{T}}{\Rightarrow} \alpha: @_{k} @_{n}\Box (\psi \to (@_{n} \langle \mathsf{F} \rangle m \to @_{m}\varphi))$ is provable in $\mathsf{T}\mathbf{EFL}$, where recall that $m$ satisfies the freshness condition. By applying Lemma \ref{lem:inv} repeatedly, we obtain the provability of $\alpha \cdot_{n} i: @_{n}\psi, \alpha \cdot_{n} i: @_{n} \langle \mathsf{F} \rangle m \stackrel{\mathcal{T} \cup \setof{\alpha \cdot_{n}i}}{\Rightarrow} \alpha\cdot_{n} i: @_{m}\varphi$ in $\mathsf{T}\mathbf{EFL}$ for some fresh label $\alpha\cdot_{n} i$. Then we proceed as follows:
\[
\infer[(@R)]{\stackrel{\mathcal{T}}{\Rightarrow} \alpha:@_{k} @_{n}\Box (\psi \to @_{n} \mathsf{F} \varphi)
}{
\infer[(\Box R)]{\stackrel{\mathcal{T}}{\Rightarrow} \alpha:@_{n}\Box (\psi \to @_{n} \mathsf{F} \varphi)}{
\infer[(\to R)]{
\stackrel{\mathcal{T} \cup \setof{\alpha \cdot_{n} i}}{\Rightarrow} \alpha \cdot_{n} i: @_{n} (\psi \to @_{n} \mathsf{F} \varphi)
}{
\infer[(@ R)]{
\alpha \cdot_{n} i: @_{n}  \psi \stackrel{\mathcal{T} \cup \setof{\alpha \cdot_{n} i}}{\Rightarrow} \alpha \cdot_{n} i: @_{n}  @_{n} \mathsf{F} \varphi
}{
\infer[(\mathsf{F}R)]{
\alpha \cdot_{n} i: @_{n}  \psi \stackrel{\mathcal{T} \cup \setof{\alpha \cdot_{n} i}}{\Rightarrow} \alpha \cdot_{n} i: @_{n} \mathsf{F} \varphi
}{
\alpha \cdot_{n} i: @_{n}\psi, \alpha \cdot_{n} i: @_{n} \langle \mathsf{F} \rangle m \stackrel{\mathcal{T} \cup \setof{\alpha \cdot_{n}i}}{\Rightarrow} \alpha\cdot_{n} i: @_{m}\varphi
}
}
}
}
},
\]
as required. 
\item[(\texttt{Nec}$_{@}$)] Let $\varphi_{j}$ $\equiv$ $@_{n}\psi$ be obtained by $(\texttt{Name})$. Fix any finite tree $\mathcal{T}$, $\alpha \in \mathcal{T}$ and nominal $k$. We show that $\stackrel{\mathcal{T}}{\Rightarrow} \alpha: @_{k}@_{n} \psi$ is provable in $\mathsf{T}\mathbf{EFL}$. By the rule $(@R)$, it suffices to establish the provability of $\stackrel{\mathcal{T}}{\Rightarrow} \alpha: @_{n} \psi$ in $\mathsf{T}\mathbf{EFL}$. This is immediate from induction hypothesis. 

\item[(\texttt{Nec}$_{\Box}$)] Let $\varphi_{j}$ $\equiv$ $\Box \psi$ be obtained by $(\texttt{Nec}_{\Box})$. Fix any finite tree $\mathcal{T}$, $\alpha \in \mathcal{T}$ and nominal $n$. By induction hypothesis, $\stackrel{\mathcal{T}\cup \setof{\alpha\cdot_{n} i}}{\Rightarrow} \alpha\cdot_{n} i : @_{n} \psi$ is provable in $\mathsf{T}\mathbf{EFL}$, where $\alpha\cdot_{n} i$ is fresh in $\mathcal{T}$. By the rule $(\Box R)$ of $\mathsf{T}\mathbf{EFL}$, the provability of $\stackrel{\mathcal{T}}{\Rightarrow} \alpha : @_{n} \Box \psi$ follows, as desired. 
\item[(\texttt{Nec}$_{\mathsf{F}}$)] Let $\varphi_{j}$ $\equiv$ $\mathsf{F} \psi$ be obtained by $(\texttt{Nec}_{\mathsf{F}})$. Fix any finite tree $\mathcal{T}$, $\alpha \in \mathcal{T}$ and nominal $n$. Let $m$ be a fresh nominal in $\psi$. By induction hypothesis, $\stackrel{\mathcal{T}}{\Rightarrow} \alpha  : @_{m} \psi$ is provable in $\mathsf{T}\mathbf{EFL}$. By the admissibility of weakening rule from Proposition \ref{prop:admissible}, we obtain the provability of $\alpha:@_{n}\langle \mathsf{F}\rangle m \stackrel{\mathcal{T}}{\Rightarrow} \alpha  : @_{m} \psi$. Since $m$ is fresh in $\psi$, the rule $(\mathsf{F}R)$ enables us to derive the provability of $\stackrel{\mathcal{T}}{\Rightarrow} \alpha  : @_{n}\mathsf{F} \psi$ in $\mathsf{T}\mathbf{EFL}$, as desired. \qed
\end{itemize}
\end{proof}

\begin{corollary}[Soudness and Completenss of $\mathsf{H}\mathbf{EFL}$]
\label{cor:summary}
The following are all equivalent: for every formula $\varphi$, 
\begin{enumerate}
\item $\varphi$ is valid in the class of all models, 
\footnote{
We do not need to assume that each of our models is {\em named} in the sense that each agent is named by an agent nominal in this statement. 
}
\item $\stackrel{\mathcal{T}}{\Rightarrow}\alpha: @_{n}\varphi$ is provable in $\mathsf{T}\mathbf{EFL}^{-}$ for all $\mathcal{T}$, $\alpha \in \mathcal{T}$ and nominals $n$,
\item $\stackrel{\mathcal{T}}{\Rightarrow} \alpha: @_{n}\varphi$ is provable in $\mathsf{T}\mathbf{EFL}$ for all $\mathcal{T}$, $\alpha \in \mathcal{T}$ and nominals $n$, 
\item $\varphi$ is provable in $\mathsf{H}\mathbf{EFL}$. 
\end{enumerate}
\end{corollary}

\begin{proof}
Item 1 is equivalent to the following:  $\stackrel{\mathcal{T}}{\Rightarrow} \alpha: @_{n}\varphi$ is true for all pairs $(\mathfrak{M}, f)$ of models and assignments, finite trees $\mathcal{T}$, $\alpha \in \mathcal{T}$ and nominals $n$. Then the equivalence between items 1, 2 and 3 holds by Corollary \ref{cor:sum_tree}. 
The direction from item 4 to item 3 holds by Theorem \ref{thm:sound_hil4efl}. 
Finally, the direction from item 3 to item 4 is established as follows. 
Suppose item 3. Let $n$ be a fresh nominal. By the supposition, 
$\stackrel{\setof{0}}{\Rightarrow} 0: @_{n}\varphi$ is provable in $\mathsf{T}\mathbf{EFL}$. It follows from Theorem \ref{thm:comp_hil4efl} that $\vdash_{\mathsf{H}\mathbf{EFL}} [\![ \stackrel{\setof{0}}{\Rightarrow} 0: @_{n}\varphi ]\!]_{0}$, which implies $\vdash_{\mathsf{H}\mathbf{EFL}} @_{n}\varphi$. 
By the axiom (\texttt{Elim}), we obtain $\vdash_{\mathsf{H}\mathbf{EFL}} n \to \varphi$ hence $\vdash_{\mathsf{H}\mathbf{EFL}} \varphi$ by $(\texttt{Name})$, as required. \qed
\end{proof}

\section{Extensions of Epistemic Logic of Friendship}
\label{sec:ext}

This section explains how we extend our tree sequent calculus $\mathsf{T}\mathbf{EFL}$ and Hilbert system $\mathsf{H}\mathbf{EFL}$. In particular, we discuss extensions where $\Box$ follows $\mathbf{S4}$ or $\mathbf{S5}$ axioms and/or the friendship relation $\asymp_{w}$ satisfies some universal properties such as irreflexivity, symmetry, etc. ($w \in W$). We note that \cite{SLG2013a,SLG2013} assume that the friendship relation $\asymp_{w}$ satisfies irreflexivity and symmetry and that $\Box$ obeys $\mathbf{S5}$ axioms. Let us introduce the following sets of additional axioms:
\begin{itemize}
\item $\mathbf{KT}$ $:=$ $\inset{\Box \varphi \to \varphi}{\varphi \in \mathsf{Form}}$.
\item $\mathbf{S4}$ $:=$ $\mathbf{KT} \cup \inset{\Box \varphi \to \Box \Box \varphi}{\varphi \in \mathsf{Form}}$.
\item $\mathbf{S5}$ $:=$ $\mathbf{S4} \cup \inset{  \varphi \to \Box \neg \Box \neg \varphi }{\varphi \in \mathsf{Form}}$. 
\end{itemize}
Let us consider formulas of the form $@_{n}m$ or $@_{n}\langle \mathsf{F} \rangle m$, which are denoted by $\rho_{i}$, $\rho_{i}'$, etc. below. Let us consider a formula $\varphi$ of the following form: 
\begin{center}
$\left( \rho_{1} \land \cdots \land \rho_{h} \right) \to \left( \rho_{1}' \lor \cdots \lor \rho_{l}' \right)$, 
\end{center}
where we note that $h$ and $l$ are possibly zero. We say that a formula of such form is a {\em regular implication}~\cite[Sec.~6]{NegriPlato2001} (we may even consider a more general class of formulas called {\em geometric formulas} (cf.~\cite{Bruennler2009}), but we restrict our attention to regular implications in this paper for simplicity).  
The corresponding frame property of a regular implication is obtained by regarding $@_{n}m$ or $@_{n}\langle \mathsf{F} \rangle m$ by ``$a_n$ = $a_m$'' and ``$a_n \asymp_{w} a_m$'' and putting the universal quantifiers for all agents and $w$. For example, irreflexivity and symmetry of 
$\asymp_{w}$ are defined by 
\begin{itemize}
\item $\mathsf{irr}_{\asymp}$ $:=$ $@_{n}\langle \mathsf{F}\rangle n \to \bot$
\item $\mathsf{sym}_{\asymp}$ $:=$ $@_{n}\langle \mathsf{F}\rangle m \to @_{m}\langle \mathsf{F}\rangle n$,
\end{itemize}
respectively.  

Now let us move to tree sequent systems. First, we introduce an inference rule for a regular implication. For a regular implication $\varphi$ displayed above, we can define the corresponding inference rule $(\mathsf{ri}(\varphi))$ for tree sequent calculus as follows (cf.~\cite{Bruennler2009}, \cite[Sec.~6]{NegriPlato2001}): 
\begin{center}
$\infer[(\mathsf{ri}(\varphi))]{\alpha:\rho_{1}, \ldots, \alpha:\rho_{h}, \Gamma \stackrel{\mathcal{T}}{\Rightarrow} \Delta }{
\alpha:\rho'_{1}, \Gamma \stackrel{\mathcal{T}}{\Rightarrow} \Delta
&
\cdots
&
\alpha:\rho'_{l}, \Gamma \stackrel{\mathcal{T}}{\Rightarrow} \Delta}
$
\end{center}
When $l$ = $0$, the rule $\mathsf{ri}(\varphi)$ is a zero premise rule of the following form:
\[
\infer[(\mathsf{ri}(\varphi))]{\alpha:\rho_{1}, \ldots, \alpha:\rho_{h}, \Gamma \stackrel{\mathcal{T}}{\Rightarrow} \Delta }{}
\]
When $\asymp_{w}$ is irreflexive or symmetric for all $w \in W$, we can obtain the following rule $(\mathsf{irr}_{\asymp})$ or $(\mathsf{sym}_{\asymp})$, respectively:
\[
\infer[(\mathsf{ri}(\mathsf{irr}_{\asymp}))]{\alpha:@_{n}\langle \mathsf{F} \rangle n, \Gamma \stackrel{\mathcal{T}}{\Rightarrow} \Delta}{}
\qquad
\infer[(\mathsf{ri}(\mathsf{sym}_{\asymp}))]{\alpha:@_{n}\langle \mathsf{F} \rangle m, \Gamma \stackrel{\mathcal{T}}{\Rightarrow} \Delta}{
\alpha: @_{m}\langle \mathsf{F} \rangle n, \Gamma \stackrel{\mathcal{T}}{\Rightarrow} \Delta
}.
\]
Let $\Lambda$ be one of $\mathbf{KT}$, $\mathbf{S4}$ and $\mathbf{S5}$ and $\Theta$ be a possibly empty finite set of regular implication schemes. 
In what follows, we define the tree sequent system $\mathsf{T}\mathbf{EFL}(\Lambda;\Theta)$. Recall that the side condition $\ddagger$ of the rule $(\Box L)$ of Table \ref{table:tree_seq_calc}. First, depending on the choice of $\Lambda$, we change the side condition $\ddagger$ of the rule $(\Box L)$ in $\mathsf{T}\mathbf{EFL}$ into the following one:
\begin{itemize}
\item {$\ddagger_{\mathbf{KT}}$: $\alpha$ $\preceq_{n}$ $\beta$, where $\preceq_{n}$ is the reflexive closure of the $n$-children relation. }
\item {$\ddagger_{\mathbf{S4}}$: $\alpha$ $\preceq_{n}^{\ast}$ $\beta$, where $\preceq_{n}^{\ast}$ is the reflexive transitive closure of the $n$-children relation. }
\item $\ddagger_{\mathbf{S5}}$: $\alpha$ $\sim_{n}$ $\beta$, where $\sim_{n}$ is the reflexive, symmetric, transitive closure of the $n$-children relation. 
\end{itemize}
When $\Lambda$ is one of $\mathbf{KT}$, $\mathbf{S4}$ and $\mathbf{S5}$, we use ``$\Lambda$'' as a subscript of the rule $(\Box L)$ as in:
\[
\infer[(\Box L_{\Lambda})]{{\alpha:@_{n}\Box \varphi}, \Gamma \stackrel{\mathcal{T}}{\Rightarrow} \Delta}{{\beta:@_{n} \varphi}, \Gamma \stackrel{\mathcal{T}}{\Rightarrow} \Delta}
\]
to indicate which side condition is considered. Second, we extend the resulting system with a finite set $\inset{(\mathsf{ri}(\varphi))}{\varphi \in \Theta}$ of inference rules, defined above, to finish to define the system $\mathsf{T}\mathbf{EFL}(\Lambda;\Theta)$. We define $\mathsf{T}\mathbf{EFL}(\Lambda;\Theta)^{-}$ as the system $\mathsf{T}\mathbf{EFL}(\Lambda;\Theta)$ without the cut rule. 

\begin{definition}
Given a set $\Psi$ of formulas and a frame $\mathfrak{F}$ = $(W, A, (R_{a})_{a \in A}, (\asymp_{w})_{w \in W})$ (a model without a valuation), we say that $\Psi$ is {\em valid} in $\mathfrak{F}$ $($notation: $\mathfrak{F} \models \Psi$$)$ if $(\mathfrak{F},V), (w,a) \models \psi$ for all $\psi \in \Psi$, valuations $V$ and pairs $(w,a) \in W\times A$. We define a class $\mathbb{M}_{\Psi}$ of models as $\inset{(\mathfrak{F},V)}{\mathfrak{F} \models \Psi}$. 
\end{definition}

\begin{theorem}
\label{thm:ex_tree}
Let $\Lambda$ be one of $\mathbf{KT}$, $\mathbf{S4}$ and $\mathbf{S5}$, and let $\Theta$ be a possibly empty finite set of regular implications. The following are all equivalent: 
\begin{enumerate}
\item $\mathfrak{M},f \models \Gamma \stackrel{\mathcal{T}}{\Rightarrow} \Delta$ for all models $\mathfrak{M} \in \mathbb{M}_{\Lambda \cup \Theta}$ and all assignments $f$. 
\item $\Gamma \stackrel{\mathcal{T}}{\Rightarrow} \Delta$ is provable in $\mathsf{T}\mathbf{EFL}(\Lambda;\Theta)^{-}$.
\item $\Gamma \stackrel{\mathcal{T}}{\Rightarrow} \Delta$ is provable in $\mathsf{T}\mathbf{EFL}(\Lambda;\Theta)$.
\end{enumerate}
Therefore, $\mathsf{T}\mathbf{EFL}(\Lambda;\Theta)$ enjoys the cut-elimination theorem. 
\end{theorem}

\begin{proof}
The direction from item 2 to item 3 is trivial and it is not difficult to establish the direction from from item 3 to item 1 (soundness result of $\mathsf{T}\mathbf{EFL}(\Lambda;\Theta)$ for the semantics). So we focus on showing the direction from item 1 to item 2 here. An outline of our proof is almost the same as in Lemma \ref{lem:saturation} and Lemma \ref{lem:truth}. First we introduce the notion of saturation of a possibly infinite tree squent as follows. As for $\varphi \equiv \left( \rho_{1} \land \cdots \land \rho_{h} \right) \to \left( \rho_{1}' \lor \cdots \lor \rho_{l}' \right) \in \Theta$, we add the following saturation condition:
\begin{description}
\item[($\mathbf{ri}(\varphi)$)] If $\alpha:\rho_{1}, \ldots, \alpha:\rho_{h} \in \Gamma$ then 
$\alpha:\rho'_{j} \in \Gamma$ for some $1 \leqslant j \leqslant l$. 
\end{description}
where the rule $(\mathsf{ri}(\varphi))$ in the tree sequent calculus is {\em not} a zero premise rule. 
Depending on our choice of $\Lambda$, we change the condition $(\Box \mathbf{l})$ as follows:
\begin{description}
\item[$(\Box \mathbf{l}_{\mathbf{KT}})$] If $\alpha:@_{n}\Box \varphi \in \Gamma$ then $\beta:@_{n}\varphi \in \Gamma$ for all $\beta$ such that $\alpha$ $\preceq_{n}$ $\beta$,
\item[$(\Box \mathbf{l}_{\mathbf{S4}})$] If $\alpha:@_{n}\Box \varphi \in \Gamma$ then $\beta:@_{n}\varphi \in \Gamma$ for all $\beta$ such that $\alpha$ $\preceq_{n}^{\ast}$ $\beta$, 
\item[$(\Box \mathbf{l}_{\mathbf{S5}})$] If $\alpha:@_{n}\Box \varphi \in \Gamma$ then $\beta:@_{n}\varphi \in \Gamma$ for all $\beta$ such that $\alpha$ $\sim_{n}$ $\beta$. 
\end{description}
Now we prove the corresponding saturation lemma to Lemma \ref{lem:saturation}.

Our proof is almost the same as in the proof of Lemma \ref{lem:saturation}. So we explain differences. For $(\textbf{Step 1})$ of the inductive step of the proof of Lemma \ref{lem:saturation}, we modify our construction as follows. Before constructing $\Gamma_{j}^\mathsf{rep}$, we construct $\Gamma_{j}^{\Theta}$ from $\Gamma_{j}$ as follows. We enumeate all the tuples in $\Gamma_{j}$ of the form $(\alpha:\rho_{1},\ldots,\alpha:\rho_{n})$ for some $\varphi \equiv \left( \rho_{1} \land \cdots \land \rho_{h} \right) \to \left( \rho_{1}' \lor \cdots \lor \rho_{l}' \right) \in \Theta$ (we note that the number of such tuples is finite). With the help of such enumeration (let $t$ be the number of such tuples), we inductively construct $(\Gamma_{j}^{(k)})_{0 \leqslant k \leqslant t}$ such that $\Gamma_{j}^{(k)} \subseteq \Gamma_{j}^{(k+1)}$ as follows. Define $\Gamma_{j}^{(0)}$ = $\Gamma_{j}$. Suppose that we have constructed $\Gamma_{j}^{(0)} \subseteq \cdots \subseteq \Gamma_{j}^{(k)}$. 
Let $k$-th tuple of the enumeration be $(\alpha:\rho_{1},\ldots,\alpha:\rho_{n})$ and the corresponding regular implication $\varphi$ is $\left( \rho_{1} \land \cdots \land \rho_{h} \right) \to \left( \rho_{1}' \lor \cdots \lor \rho_{l}' \right)$. We can find some index $f$ such that $\alpha: \rho_{f}', \Gamma_{j}^{(k)} \stackrel{\mathcal{T}}{\Rightarrow} \Delta$ is unprovable in $\mathsf{T}\mathbf{EFL}(\Lambda;\Theta)^{-}$ by the rule $(\mathsf{ri}(\varphi))$. Then we define $\Gamma_{j}^{(k+1)}$ as $\alpha: \rho_{f}', \Gamma_{j}^{(k)}$. Finally we define $\Gamma_{j}^{\Theta}$ := $\bigcup_{1 \leqslant k \leqslant t} \Gamma_{j}^{(k)}$. Then we do the same construction as in Step 1 for $\Gamma_{j}^{\Theta}$ instead of $\Gamma_{j}$. For Step 2, there is no substantial change. This finishes to establish the corresponding saturation lemma to Lemma \ref{lem:saturation}. 

Next we comment on the corresponding lemma to Lemma \ref{lem:truth}. Let $\Gamma \stackrel{\mathcal{T}}{\Rightarrow} \Delta$ be a saturated and unprovable tree sequent in $\mathsf{T}\mathbf{EFL}(\Lambda;\Theta)^{-}$. As in the statement of Lemma \ref{lem:truth}, we define the derived model $\mathfrak{M}$ in the same way except $R_{|n|}$. Depending on our choice of $\Lambda$, we define $R_{|n|}$ as follows:
\begin{description}
\item[$(\mathbf{KT})$] $\alpha R_{|n|} \beta$ iff $\alpha$ $\preceq_{m}$ $\beta$ for some $m \in |n|$.
\item[$(\mathbf{S4})$] $\alpha R_{|n|} \beta$ iff $\alpha$ $\preceq_{m}^{\ast}$ $\beta$ for some $m \in |n|$.
\item[$(\mathbf{S5})$] $\alpha R_{|n|} \beta$ iff $\alpha$ $\sim_{m}$ $\beta$ for some $m \in |n|$.\end{description}
Then it is easy to see $R_{|n|}$ satisfies the corresponding properties of $\Lambda$, i.e., $R_{|n|}$ is reflexive when $\Lambda$ is $\mathbf{KT}$, $R_{|n|}$ is a pre-order when $\Lambda$ is $\mathbf{S4}$, $R_{|n|}$ is an equivalence relation when $\Lambda$ is $\mathbf{S5}$. The remaining argument is the same as in the proof of Lemma \ref{lem:truth}. Moreover, it follows from the saturation condition $\mathbf{ri}(\varphi)$ and the unprovability of $\Gamma \stackrel{\mathcal{T}}{\Rightarrow} \Delta$ in $\mathsf{T}\mathbf{EFL}(\Lambda;\Theta)^{-}$ that the corresponding properties of $\Theta$ are satisfied. This enables us to conclude the derived model $\mathfrak{M}$ belongs to $\mathbb{M}_{\Lambda \cup \Theta}$. This finishes showing the direction from item 1 to item 2. 
\qed
\end{proof}


\begin{definition}
When $\Lambda$ is one of $\mathbf{KT}$, $\mathbf{S4}$ and $\mathbf{S5}$ and $\Theta$ is a finite set of regular implications, a Hilbert system $\mathsf{H}\mathbf{EFL}(\Lambda \cup \Theta)$ is defined as the axiomatic extension of $\mathsf{H}\mathbf{EFL}$ by new axioms $\Lambda \cup \Theta$. 
\end{definition}


\begin{theorem}
\label{thm:ex_hil}
Let $\Lambda$ be one of $\mathbf{KT}$, $\mathbf{S4}$ and $\mathbf{S5}$, and let $\Theta$ be a possibly empty finite set of regular implications. The following are all equivalent: for every formula $\varphi$, 
\begin{enumerate}
\item $\varphi$ is valid in $\mathbb{M}_{\Lambda \cup \Theta}$. 
\item $\stackrel{\mathcal{T}}{\Rightarrow}\alpha: @_{n}\varphi$ is provable in $\mathsf{T}\mathbf{EFL}(\Lambda \cup \Theta)^{-}$ for all $\mathcal{T}$, $\alpha \in \mathcal{T}$ and nominals $n$,
\item $\stackrel{\mathcal{T}}{\Rightarrow} \alpha: @_{n}\varphi$ is provable in $\mathsf{T}\mathbf{EFL}(\Lambda \cup \Theta)$ for all $\mathcal{T}$, $\alpha \in \mathcal{T}$ and nominals $n$, 
\item $\varphi$ is provable in $\mathsf{H}\mathbf{EFL}(\Lambda \cup \Theta)$. 
\end{enumerate}
\end{theorem}

\begin{proof}
By Theorem \ref{thm:ex_tree}, we can establish the equivalence between items 1, 2 and 3. 
We are going to provide our argument for a direction from item 4 to item 3 and a direction from item 3 to item 4.

From item 4 to item 3, we prove a similar statement to Theorem \ref{thm:sound_hil4efl}. But it suffices to prove the additional axioms from $\Lambda \cup \Theta$ are provable in $\mathsf{T}\mathbf{EFL}(\Lambda \cup \Theta)$. In what folows, let us fix any tree $\mathcal{T}$, $\alpha \in \mathcal{T}$ and nominal $n$. First of all, let $\psi \equiv \left( \rho_{1} \land \cdots \land \rho_{h} \right) \to \left( \rho_{1}' \lor \cdots \lor \rho_{l}' \right) \in \Theta$. We show $\stackrel{\mathcal{T}}{\Rightarrow} \alpha: @_{n}\psi$ is provable in $\mathsf{T}\mathbf{EFL}(\Lambda \cup \Theta)$. The crucial part of this derivation is the following:
\[
\infer[(\mathsf{ri}(\varphi))]{
\alpha: \rho_{1}, \ldots, \alpha: \rho_{h} \stackrel{\mathcal{T}}{\Rightarrow} \alpha: \rho_{1}', \ldots, \alpha: \rho_{l}'
}{
\infer[(\mathsf{id})]{
\alpha: \rho_{1}' \stackrel{\mathcal{T}}{\Rightarrow} \alpha: \rho_{1}', \ldots, \alpha: \rho_{l}'}{}
&
\cdots
&
\infer[(\mathsf{id})]{
\alpha: \rho_{l}' \stackrel{\mathcal{T}}{\Rightarrow} \alpha: \rho_{1}', \ldots, \alpha: \rho_{l}'}{}
}.
\]
Let us move to $\Lambda$. When $\Lambda$ is $\mathbf{KT}$, it suffices to give the following derivation: 
\[
\infer[ (\Box L_{\mathbf{KT}} ) ]{ \alpha: @_{n}\Box \varphi \stackrel{\mathcal{T}}{\Rightarrow}  \alpha: @_{n}\varphi}{
\infer[(\mathsf{id})]{\alpha: @_{n} \varphi \stackrel{\mathcal{T}}{\Rightarrow}  \alpha: @_{n}\varphi}{}
}.
\]
When $\Lambda$ is $\mathbf{K4}$, it suffices to give the following: 
\[
\infer[(\Box R)]{\alpha: @_{n}\Box \varphi \stackrel{\mathcal{T}}{\Rightarrow}  \alpha: @_{n}\Box \Box \varphi }{
\infer[(\Box R )]{  \alpha: @_{n}\Box \varphi \stackrel{\mathcal{T} \cup \setof{\alpha \cdot_{n} i}}{\Rightarrow}  \alpha \cdot_{n} i : @_{n} \Box \varphi }{
\infer[(\Box L_{\mathbf{S4}})]{ \alpha: @_{n}\Box \varphi \stackrel{\mathcal{T} \cup \setof{\alpha \cdot_{n} i, \alpha \cdot_{n} i \cdot_{n} j}}{\Rightarrow}  \alpha \cdot_{n} i \cdot_{n} j: @_{n} \varphi }{
\infer[(\mathsf{id})]{
\alpha \cdot_{n} i \cdot_{n} j: @_{n} \varphi \stackrel{\mathcal{T} \cup \setof{\alpha \cdot_{n} i, \alpha \cdot_{n} i \cdot_{n} j}}{\Rightarrow}  \alpha \cdot_{n} i \cdot_{n} j: @_{n} \varphi
}{}
}
}
}.
\]
Finally, if $\Lambda$ is $\mathbf{S5}$, on the top of the above two derivation, it suffices to consider the following derivation: 
\[
\infer[(\Box R)]{ \alpha: @_{n} \varphi \stackrel{\mathcal{T}}{\Rightarrow}  \alpha: @_{n}\Box \neg \Box \neg \varphi  }{
\infer[(\neg R)]{\alpha: @_{n} \varphi \stackrel{\mathcal{T}\cup \setof{\alpha \cdot_{n} i}}{\Rightarrow} \alpha \cdot_{n} i: @_{n}\neg \Box \neg \varphi  }{
\infer[(\Box L_{\mathbf{S5}})]{\alpha \cdot_{n} i: @_{n}\Box \neg \varphi, \alpha: @_{n} \varphi \stackrel{\mathcal{T}\cup \setof{\alpha \cdot_{n} i}}{\Rightarrow} }{
\infer[(\neg L)]{\alpha : @_{n} \neg \varphi, \alpha: @_{n} \varphi \stackrel{\mathcal{T}\cup \setof{\alpha \cdot_{n} i}}{\Rightarrow}}{
\infer[(\mathsf{id})]{
\alpha: @_{n} \varphi \stackrel{\mathcal{T}\cup \setof{\alpha \cdot_{n} i}}{\Rightarrow} \alpha : @_{n} \varphi}{}
}
}
}
}.
\]
For the direction from item 3 to item 4, it suffices to establish the formulaic translation of the rules $(\mathsf{ri}(\varphi))$ for all $\varphi \in \Theta$, $(\Box L_{\mathbf{KT}})$, $(\Box L_{\mathbf{S4}})$ and $(\Box L_{\mathbf{S5}})$ at the root node preserves the provability in the corresponding system $\mathsf{H}\mathbf{EFL}(\Lambda \cup \Theta)$. Since the case of $(\mathsf{ri}(\varphi))$ is not so difficult for every $\varphi \in \Theta$, we focus on all the other rules. All the other rules have the following form:
\[
\infer[(\Box L_{\Lambda})]{\alpha:@_{n}\Box \varphi, \Gamma \stackrel{\mathcal{T}}{\Rightarrow} \Delta  }{\beta:@_{n} \varphi, \Gamma \stackrel{\mathcal{T}}{\Rightarrow} \Delta}
\]
Let us suppose that $i$ be the root node of $\mathcal{T}$. Let $(\alpha_{0}, \ldots, \alpha_{l})$ be the unique path from the $\alpha$ ($\equiv \alpha_{l}$) to the root node $i$ ($\equiv \alpha_{0}$). Similarly as in the proof of the case of the rule $(\Box L)$ in Theorem \ref{thm:comp_hil4efl}, by induction on $h$, that: 
\[
\vdash_{\mathsf{H}\mathbf{EFL}(\Lambda \cup \Theta)}  [\![ \beta:@_{n} \varphi, \Gamma \stackrel{\mathcal{T}}{\Rightarrow} \Delta ]\!]_{\alpha_{l-h}} \to [\![ \alpha:@_{n}\Box \varphi, \Gamma \stackrel{\mathcal{T}}{\Rightarrow} \Delta ]\!]_{\alpha_{l-h}},
\]
where we often omit the subscript of $\vdash_{\mathsf{H}\mathbf{EFL}(\Lambda \cup \Theta)}$ to simply write $\vdash$ below when no confusion arises. 
When $\Lambda$ is $\mathbf{KT}$, it suffice to check an application of $(\Box L_{\mathbf{KT}})$ where $\beta$ is $\alpha$ itself. We only establish the base case where $h$ = $0$. That is, we establish\[
\vdash_{\mathsf{H}\mathbf{EFL}(\mathbf{KT} \cup \Theta)}  [\![ \alpha:@_{n} \varphi, \Gamma \stackrel{\mathcal{T}}{\Rightarrow} \Delta ]\!]_{\alpha} \to [\![ \alpha:@_{n}\Box \varphi, \Gamma \stackrel{\mathcal{T}}{\Rightarrow} \Delta ]\!]_{\alpha},
\]
To show this, it suffice to show the following:
\[
\vdash_{\mathsf{H}\mathbf{EFL}(\mathbf{KT} \cup \Theta)} ((\gamma \land @_{n} \varphi) \to \delta)  \to ((\gamma \land @_{n} \Box \varphi) \to \delta),
\]
which is easily obtained by $\vdash_{\mathsf{H}\mathbf{EFL}(\mathbf{KT} \cup \Theta)} \Box \varphi \to \varphi$, $(\texttt{Nec}_{@})$ and $(\texttt{K}_{@})$. 

Let us move to the case where $\Lambda$ is $\mathbf{S4}$. In this case, we suffice to check an application of $(\Box L_{\mathbf{S4}})$ where $\beta$ is a grand $n$-child of $\alpha$, i.e., $\beta$ is an $n$-child of $\alpha'$ and $\alpha'$ is an $n$-child of $\alpha'$ for some label $\alpha' \in \mathcal{T}$. To show the base case where $h$ $=$ 0, it suffice to prove the following: 
\begin{align*}
\vdash_{\mathsf{H}\mathbf{EFL}(\mathbf{S4} \cup \Theta)} 
&\left(\gamma_{\alpha}  \to (\delta_{\alpha} \lor @_{n}(\gamma_{\alpha'}  \to (\delta_{\alpha'} \lor @_{n} \Box (  (@_{n}\varphi \land \gamma_{\beta}) \to \delta_{\beta} ) ))) \right)\\
&\to
\left((@_{n}\Box \varphi \land \gamma_{\alpha}) \to (\delta_{\alpha} \lor @_{n}(\gamma_{\alpha'}  \to (\delta_{\alpha'} \lor @_{n} \Box (  \gamma_{\beta} \to \delta_{\beta} ) ))) \right),
\end{align*}
which is provable by $\vdash_{\mathsf{H}\mathbf{EFL}(\mathbf{S4} \cup \Theta)} \Box \varphi \to \Box\Box\varphi$, $(\texttt{Nec}_{@})$ and $(\texttt{K}_{@})$.

Finally, if $\Lambda$ is $\mathbf{S5}$, we suffice to check an application of $(\Box L_{\mathbf{S5}})$ where $\beta$ is an {\em $n$-parent} of $\alpha$, i.e., $\alpha$ is an $n$-child of $\beta$. In this case, our base case is $h$ = $1$, i.e., we show the preservation of the provability of the formulaic translation at $\beta$: 
\[
\vdash_{\mathsf{H}\mathbf{EFL}(\mathbf{S5} \cup \Theta)}  [\![ \beta:@_{n} \varphi, \Gamma \stackrel{\mathcal{T}}{\Rightarrow} \Delta ]\!]_{\beta} \to [\![ \alpha:@_{n}\Box \varphi, \Gamma \stackrel{\mathcal{T}}{\Rightarrow} \Delta ]\!]_{\beta}.
\]
To show it, it suffices to establish the following:
\begin{align*}
\vdash_{\mathsf{H}\mathbf{EFL}(\mathbf{S5} \cup \Theta)} 
&\left((@_{n}\varphi \land \gamma_{\beta}) \to (\delta_{\beta} \lor @_{n}\Box(\gamma_{\alpha} \to \delta_{\alpha})) \right)\\
&\to
\left(\gamma_{\beta} \to (\delta_{\beta} \lor @_{n}\Box((@_{n}\Box \varphi  \land \gamma_{\alpha}) \to \delta_{\alpha})) \right).
\end{align*}
Since this is equivalent with: 
\begin{align*}
\vdash_{\mathsf{H}\mathbf{EFL}(\mathbf{S5} \cup \Theta)} 
&\left( \gamma_{\beta} \to ( \neg @_{n} \varphi \lor  \delta_{\beta} \lor @_{n}\Box(\gamma_{\alpha} \to \delta_{\alpha})) \right)\\
&\to
\left(\gamma_{\beta} \to (\delta_{\beta} \lor @_{n}\Box((@_{n}\Box \varphi  \land \gamma_{\alpha}) \to \delta_{\alpha})) \right),
\end{align*}
we need to establish:
\[
\vdash_{\mathsf{H}\mathbf{EFL}(\mathbf{S5} \cup \Theta)} @_{n}\neg \varphi \to @_{n}\Box @_{n} \neg  \Box  \varphi
\]
by $\vdash \neg @_{n} \psi \leftrightarrow @_{n}\neg \psi$. With the help of the axiom (\texttt{DCom}$@\Box$), $(\texttt{Nec}_{@})$ and $(\texttt{K}_{@})$, the provability above is reduced to 
$\vdash_{\mathsf{H}\mathbf{EFL}(\mathbf{S5} \cup \Theta)} \neg \varphi \to \Box \neg  \Box  \varphi$, which is easily obtained from the axiom scheme $\psi \to \Box \neg \Box \neg \psi$. \qed
\end{proof}

\noindent Recall that \cite{SLG2013a,SLG2013} assume that the friendship relation $\asymp_{w}$ satisfies irreflexivity and symmetry and that $\Box$ obeys $\mathbf{S5}$ axioms. As a corollary of Theorem \ref{thm:ex_hil}, the following provides a complete axiomatization of the logic studied in~\cite{SLG2013a,SLG2013}, where $\mathsf{irr}_{\asymp}$ is $@_{n}\langle \mathsf{F}\rangle n \to \bot$ and $\mathsf{sym}_{\asymp}$ is $@_{n}\langle \mathsf{F}\rangle m \to @_{m}\langle \mathsf{F}\rangle n$.

\begin{corollary}
\label{cor:comp_efls5}
The following are all equivalent: for every formula $\varphi$, 
\begin{enumerate}
\item $\varphi$ is valid in $\mathbb{M}_{\mathbf{S5} \cup \setof{\mathsf{irr}_{\asymp}, \mathsf{sym}_{\asymp}}}$. 
\item $\stackrel{\mathcal{T}}{\Rightarrow}\alpha: @_{n}\varphi$ is provable in $\mathsf{T}\mathbf{EFL}(\mathbf{S5} \cup \setof{\mathsf{irr}_{\asymp}, \mathsf{sym}_{\asymp}})^{-}$ for all finite tree $\mathcal{T}$, $\alpha \in \mathcal{T}$ and nominals $n$,
\item $\stackrel{\mathcal{T}}{\Rightarrow} \alpha: @_{n}\varphi$ is provable in $\mathsf{T}\mathbf{EFL}(\mathbf{S5} \cup \setof{\mathsf{irr}_{\asymp}, \mathsf{sym}_{\asymp}})$ for all finite tree  $\mathcal{T}$, $\alpha \in \mathcal{T}$ and nominals $n$, 
\item $\varphi$ is provable in $\mathsf{H}\mathbf{EFL}(\mathbf{S5} \cup \setof{\mathsf{irr}_{\asymp}, \mathsf{sym}_{\asymp}})$. 
\end{enumerate}
\end{corollary}

\section{Further Directions}
\label{sec:fd}
This paper positively answered the question if the set of all valid formulas of $\mathbf{EFL}$ in the class of all models is axiomatizable. We list some directions for further research. 
\begin{enumerate}
\item Is $\mathsf{H}\mathbf{EFL}$ or $\mathsf{T}\mathbf{EFL}$ decidable?
\item Is it possible to provide a syntactic proof of the cut elimination theorem of $\mathsf{T}\mathbf{EFL}$?
\item Can we reformulate our sequent calculus into a G3-style calculus, i.e., a contraction-free calculus, all of whose rules are height-preserving invertible?
\item Provide a G3-style labelled sequent calculus for $\mathbf{EFL}$ based on the idea of doubly labelled formula $(x,y):\varphi$. This is an extension of G3-style labelled sequent calculus for modal logic in~\cite{NegriPlato2011,Negri2005}.
\item Prove the semantic completeness of $\mathsf{H}\mathbf{EFL}$ and its extensions by specifying the notion of {\em canonical model}.
\item Can we apply our technique of this paper to obtain a Hilbert system of {\em Term Modal Logics} which is proposed in~\cite{Fitting2001}?
\footnote{I would like to thank the anonymous reviewers of LORI VI for their careful reading of the manuscript and their many useful comments and suggestions.  I presented the contents of this paper first at the 48th MLG meeting at Kaga, Ishikawa, Japan on 6th December 2013 and then at Kanazawa Workshop for Epistemic Logic and its Dynamic Extensions, Kanazawa, Japan on 22nd February 2014. I would like to thank Alexandru Baltag, Jeremy Seligman and Fenrong Liu for fruitful discussions of the topic. All errors, however, are mine. The work of the author was partially supported by JSPS KAKENHI Grant-in-Aid for Young Scientists (B) Grant Number 15K21025 and Grant-in-Aid for Scientific Research (B) Grant Number 17H02258, and JSPS Core-to-Core Program (A. Advanced Research Networks). 
}
\end{enumerate}

\bibliographystyle{plain}

\end{document}